\documentclass[11pt,english]{article}
\usepackage[T1]{fontenc}
\usepackage[latin9]{inputenc}
\usepackage{color}
\usepackage{babel}
\usepackage{float}
\usepackage{amsmath}
\usepackage{amsthm}
\usepackage{amssymb}
\usepackage{graphicx}
\PassOptionsToPackage{normalem}{ulem}
\usepackage{ulem}
\usepackage[unicode=true,
 bookmarks=false,
 breaklinks=false,pdfborder={0 0 1},backref=section,colorlinks=true]
 {hyperref}
\hypersetup{
 pdfauthor={Yuri Ximenes Martins},
 linkcolor=blue,citecolor=blue,filecolor=blue,pagecolor=blue,urlcolor=red}

\makeatletter

\providecommand{\tabularnewline}{\\}

\theoremstyle{plain}
\newtheorem{prop}{\protect\propositionname}[section]
\theoremstyle{definition}
\newtheorem{example}{\protect\examplename}[section]
\theoremstyle{remark}
\newtheorem{rem}{\protect\remarkname}[section]
\theoremstyle{plain}
\newtheorem{lem}{\protect\lemmaname}[section]
\theoremstyle{plain}
\newtheorem{thm}{\protect\theoremname}[section]
\theoremstyle{plain}
\newtheorem{cor}{\protect\corollaryname}[section]
\theoremstyle{plain}
\newtheorem*{conjecture*}{\protect\conjecturename}

\usepackage[all,2cell]{xy}  \UseAllTwocells \SilentMatrices
\usepackage{hyperref}
\usepackage{amsmath}
\usepackage{slashed}
\usepackage{graphicx}

\usepackage[misc,geometry]{ifsym} 

\hypersetup{colorlinks=true, linkcolor=blue, citecolor=blue, filecolor=blue, pagecolor=blue, urlcolor=red, pdfauthor={}, pdftitle={geometric_regularity_1}}

\makeatother

\providecommand{\conjecturename}{Conjecture}
\providecommand{\corollaryname}{Corollary}
\providecommand{\examplename}{Example}
\providecommand{\lemmaname}{Lemma}
\providecommand{\propositionname}{Proposition}
\providecommand{\remarkname}{Remark}
\providecommand{\theoremname}{Theorem}

\begin{document}
\title{Geometric Obstructions on Gravity}
\author{Yuri Ximenes Martins\footnote{yurixm@ufmg.br (corresponding author)}\,  and  Rodney Josu\'e Biezuner\footnote{rodneyjb@ufmg.br} }
\maketitle
\noindent \begin{center}
\textit{Departamento de Matem\'atica, ICEx, Universidade Federal de Minas Gerais,}  \\  \textit{Av. Ant\^onio Carlos 6627, Pampulha, CP 702, CEP 31270-901, Belo Horizonte, MG, Brazil}
\par\end{center}
\begin{abstract}
These are notes for a short course and some talks gave at Departament
of Mathematics and at Departament of Physics of Federal University
of Minas Gerais, based on the author's paper \cite{EHP_eu}. Some
new information and results are also presented. Unlike the original
work, here we try to give a more physical emphasis. In this sense,
we present obstructions to realize gravity, modeled by the tetradic
Einstein-Hilbert-Palatini (EHP) action functional, in a general geometric
setting. In particular, we show that if spacetime has dimension $n\geq4$,
then the cosmological constant plays no role in any ``concrete geometries''
other than Lorentzian. If $n\geq6$, then the entire EHP theory is
trivial, meaning that Lorentzian geometry is (essentially) the only
``concrete geometry'' in which gravity (i.e, the EHP action functional)
makes sense. Examples of ``concrete geometries'' include those locally
modeled by group reductions $H\hookrightarrow O(k;A)$ for some $k$
and some algebra $A$, so that Riemannian geometry, Hermitian geometry,
Kähler geometry and symplectic geometry, as well as Type II geometry,
Hitchin's generalized complex geometry and $G_{2}$-geometry are included.
We also study EHP theory in ``abstract geometries'', such as graded
geometry (and hence supergeometry), and we show how the obstruction
results extend to this context. We construct two theories naturally
associated to EHP, which we call the geometric/algebraic dual of EHP,
and we analyze the effect of the obstructions in these new theories.
Finally, we speculate (and provide evidence for) the existence of
a ``universal obstruction condition''.
\end{abstract}

\section{Introduction}

Einstein's theory of gravity is formulated in Lorentzian geometry:
spacetime is regarded as an orientable four-dimensional Lorentzian
manifold $(M,g)$ whose Lorentzian metric $g$ is a critical point
of the Einstein-Hilbert action functional 
\begin{equation}
S_{EH}[g]=\int_{M}(R_{g}-2\Lambda)\cdot\omega_{g},\label{einstein_hilbert_action}
\end{equation}
defined on the space of all possible Lorentzian metrics in $M$. Here,
$\Lambda\in\mathbb{R}$ is a parameter (the cosmological constant),
while $R_{g}$ and $\omega_{g}$ are, respectively, the scalar curvature
and the volume-form of $g$, locally written as $g^{ij}\operatorname{Ric}_{ij}$
and $\sqrt{\vert\det g\vert}\,dx^{1}\wedge...\wedge dx^{4}$, where
$\operatorname{Ric}_{ij}$ are the components of the Ricci tensor.
Such critical points are precisely those satisfying the vacuum Einstein's
equation 
\begin{equation}
\operatorname{Ric}_{ij}-\frac{1}{2}R_{g}g_{ij}+\Lambda g_{ij}=0.\label{einstein_equation}
\end{equation}

The action (\ref{einstein_hilbert_action}) is about tensors, so that
we can say that General Relativity is generally formulated in a ``tensorial
approach''. Electromagnetism (or, more generally, Yang-Mills theories)
was also initially formulated in a ``tensorial approach''. Indeed,
the action functional for electromagnetism is given by 
\[
S_{YM}[A]=-\frac{1}{4}\int_{M}F_{ij}\tilde{F}^{ij},
\]
where $F$ and $\tilde{F}$ are the ``electromagnetic tensor'' and
the ``dual eletromagnetic tensor'', with components 
\[
F_{ij}=\partial_{i}A_{j}-\partial_{j}A_{i}\quad\text{and \quad}\tilde{F}^{ij}=\frac{1}{2}\epsilon^{ijkl}F_{kl}\sqrt{\vert\det g\vert},
\]
respectively. Varying this action we get Maxwell's equations in their
tensorial formulation: 
\begin{equation}
\partial_{k}\tilde{F}^{kl}+A_{k}\tilde{F}^{kl}=0.\label{inhomogeneous_maxwell}
\end{equation}

Notice that this is actually \textbf{one half} of the actual Maxwell's
equations, corresponding to the ``inhomogeneous equations''. In
order to get the full equations we need to take into account an external
set of equations, corresponding to the ``homogeneous part'' of Maxwell's
equations: 
\begin{equation}
\epsilon^{ijklmn}\partial_{l}F_{mn}=0.\label{homogeneous_maxwell}
\end{equation}

It happens that, when we move from the tensorial language to the language
of differential forms and connections on bundles, we rediscover (\ref{homogeneous_maxwell})
as a\emph{ geometric identity} (Bianchi identity), so that (\ref{homogeneous_maxwell})
actually holds \emph{a priori} and not \emph{a posteriori}, as was
suggested using tensorial language. This is not a special feature
of Yang-Mills theories. In fact, General Relativity (GR) also manifests
this type of behavior: alone, Einstein's equation (\ref{einstein_equation})
does not completely specify a system of (GR); we also need to assume
that the connection in question is precisely the Levi-Civita connection
of $g$. But, as in the case of Bianchi identity for Yang-Mills theories,
this assumption can be avoided if we use the language of differential
forms.

Indeed, in the so-called \emph{first order formulation of gravity}
(also known as \emph{tetradic gravity}) we can rewrite Einstein-Hilbert
action functional as the \emph{Einstein-Hilbert-Palatini }(EHP)\emph{
action}: 
\begin{equation}
S_{EHP}[e,\omega]=\int_{M}\operatorname{tr}(e\curlywedge_{\rtimes}e\curlywedge_{\rtimes}\Omega+\frac{\Lambda}{6}e\curlywedge_{\rtimes}e\curlywedge_{\rtimes}e\curlywedge_{\rtimes}e),\label{EHP_action}
\end{equation}
where $e$ and $\omega$ are 1-forms in the frame bundle $FM$ with
values in the Lie algebras $\mathbb{R}^{3,1}$ and $\mathfrak{o}(3,1)$,
called \emph{tetrad} and \emph{spin connection}, respectively, and
$\curlywedge_{\rtimes}$ is a type of ``wedge product''\footnote{In this paper we will work with many different types of ``wedge products'',
satisfying very different properties. Therefore, in order to avoid
confusion, we will not follow the literature, but introduce specific
symbols for each of them.} induced by matrix multiplication in $O(3,1)$. Because $M$ is Lorentzian,
its frame bundle is structured over $O(3,1)$ and the spin connection
is actually a connection in this bundle. The tetrad $e$ is such that
for every $a\in FM$ the corresponding map $e_{a}:FM_{a}\rightarrow\mathbb{R}^{3,1}$
is an isomorphism, so that it has the geometrical meaning of a soldering
form. Furthermore, in (\ref{EHP_action}) we have a 2-form $\Omega$
with values in $\mathfrak{o}(3,1)$, representing the curvature $\Omega=d\omega+\omega\curlywedge_{\rtimes}\omega$
of $\omega$. The translational algebra and the Lorentz algebra fits
into the Poincaré Lie algebra $\mathfrak{iso}(3,1)=\mathbb{R}^{3,1}\rtimes\mathfrak{o}(3,1)$.
The equations of motion for the EHP action are \cite{PhD_cartan_connections}
\[
d_{\omega}e\curlywedge_{\rtimes}e=0\quad\text{and}\quad e\curlywedge_{\rtimes}\Omega+\frac{\Lambda}{6}e\curlywedge_{\rtimes}e\curlywedge_{\rtimes}e=0.
\]
The second of them is just Einstein's equation (\ref{einstein_equation})
rewritten in the language of forms, while the first, due to the fact
that $e$ is an isomorphism, is equivalent to $d_{\omega}e=0$. As
the $2$-form $\Theta=d_{\omega}e$ describes precisely the torsion
of the spin connection $\omega$, the first equation of motion implies
that $\omega$ is actually the only torsion-free connection compatible
with the metric: the \emph{Levi-Civita connection}.

These examples emphasize that the language of differential forms seems
to be a nice way to describe physical theories. Indeed, \cite{hitchin_functional_1,hitchin_functional_2}
started a program that attempts to unify different physical theories
by using the same type of functional on forms, the so-called \emph{Hitchin's
functional}. Furthermore, in \cite{topological_M_theory} it was shown
that all known ``gravity theories defined by forms'' can really
be unified in that they are particular cases of a single \emph{topological
M-theory}. Motivated by this philosophy, in this article we will consider
generalizations of classical EHP theories (\ref{EHP_action}).

First of all notice that if the underlying manifold is $n$-dimensional
spacetime (instead of four-dimensional), (\ref{EHP_action}) can be
immediately generalized by considering $\mathfrak{iso}(n-1,1)$-valued
forms and an action functional given by 
\begin{equation}
S_{n,\Lambda}[e,\omega]=\int_{M}\operatorname{tr}(e\curlywedge_{\rtimes}...\curlywedge_{\rtimes}e\curlywedge_{\rtimes}\Omega+\frac{\Lambda}{(n-1)!}e\curlywedge_{\rtimes}...\curlywedge_{\rtimes}e).\label{EHP_action-1}
\end{equation}
However, this is not the only generalization that can be considered.
We notice that in the modern language of differential geometry, the
data defined by the pairs $(\omega,e)$ above corresponds to \emph{reductive
Cartan connections} on the frame bundle $FM$ with respect to the
inclusion $O(n-1,1)\hookrightarrow\operatorname{Iso}(n-1,1)$ of the
Lorentz group into the Poincaré group. Indeed, given a $G$-bundle
$P\rightarrow M$ and a structural group reduction $H\hookrightarrow G$,
we recall that a \emph{Cartan connection} in $P$ for such a reduction
is a $G$-connection $\nabla$ in $P$ which projects isomorphically
onto $\mathfrak{g}/\mathfrak{h}$ in each point. Explicitly, it is
an (horizontal and equivariant) $\mathfrak{g}$-valued $1$-form $\nabla:TP\rightarrow\mathfrak{g}$
such that in each $a\in P$ the composition below is an isomorphism
of vector spaces: $$
\xymatrix{\ar@/_{0.4cm}/[rr]_{\simeq}TP_{a}\ar[r]^{\nabla_{a}}&\mathfrak{g}\ar[r]^{\pi}&\mathfrak{g/h}}
$$Due to the decomposition (as vector spaces) $\mathfrak{g}\simeq\mathfrak{g}/\mathfrak{h}\oplus\mathfrak{h}$,
the existence of the isomorphism above allows us to write $\nabla_{a}=e_{a}+\omega_{a}$
in each point. If this varies smoothly we say that the Cartan connection
is \emph{reductive} (or \emph{decomposable}).

Now, looking at (\ref{EHP_action-1}) we see that the groups $O(n-1,1)\hookrightarrow\operatorname{Iso}(n-1,1)$
and the bundle $FM$ do not appear explicitly, so that we can think
of considering an analogous action for other group reductions $H\hookrightarrow G$
in other bundle $P$. Intuitively, this means that we are trying to
realize gravity (as modeled by EHP theory) in different geometries
other than Lorentzian. But, \emph{is }(\ref{EHP_action-1}) \emph{always
nontrivial? }In other words,\emph{ is it possible to realize gravity
in any geometry?} Notice that it is natural to expect that the algebraic
properties of $H$ and $G$ will be related to the fundamental properties
of the corresponding version of (\ref{EHP_action-1}), so that a priori
there should exist some abstract algebraic conditions under which
(\ref{EHP_action-1}) is trivial. In this article we will search for
such nontrivial conditions, which we call \emph{geometric obstructions}.
Thus, one can say that \emph{we will do Geometric Obstruction Theory
applied to EHP theory}. For instance, one of the results that we will
prove is the following, which emphasizes that some notion of ``solvability''
is crucial:$\underset{\underset{\;}{\;}}{\;}$

\noindent \textbf{Theorem A.} Let $M$ be an $n$-dimensional spacetime
and $P\rightarrow M$ be a $\mathbb{R}^{k}\rtimes H$-bundle, endowed
with the group reduction $H\hookrightarrow\mathbb{R}^{k}\rtimes H$.
If $\mathfrak{h}$ is a $(k,s)$-solv algebra and $n\geq k+s+1$,
then the cosmological constant plays no role; if $n\geq k+s+3$, then
the entire EHP theory is trivial.$\underset{\underset{\;}{\;}}{\;}$

This type of obstruction theorem gives restrictions only on the dimension
of spacetime. We also show that if we restrict to torsion-free connections,
then there are nontrivial geometric obstructions that gives restrictions
not only on the dimension, but also on the topology of spacetime.
For instance,$\underset{\underset{\;}{\;}}{\;}$

\noindent \textbf{Theorem B.} Let $M$ be an $n$-dimension Berger
manifold\emph{ }endowed with an $H$-structure.\emph{ }If $\mathfrak{h}\subset\mathfrak{so}(n)$,
then the torsionless EHP theory is nontrivial only if $n=2,4$ and
$M$ is Kähler. In particular, if $M$ is compact and $H^{2}(M;\mathbb{C})=0$,
then it must be a K3-surface.$\underset{\underset{\;}{\;}}{\;}$

The generalization of EHP from Lorentzian geometry to arbitrary geometry
is not the final step. Indeed, looking at (\ref{EHP_action-1}) again
we see that it remains well defined if one forgets that $e$ and $\omega$
together define a Cartan connection. In other words, all one needs
is the fact they are 1-forms that take values in some algebra. This
leads us to consider some kind of \emph{algebra-valued EHP theories},
defined on certain algebra-valued differential forms, for which we
show that Theorem A remains valid almost \emph{ipsis litteris}. We
also show that if the algebra in question is endowed with a grading,
then the ``solvability condition'' in Theorem A can be weakened.

This paper is organized as follows: in Section \ref{sec_lie_alg}
we review some facts concerning algebra-valued differential forms
and we study the ``solvability conditions'' that will appear in
the obstructions results. In Section \ref{sec_obstruction} we introduce
EHP theory into the ``linear''/ ``extended-linear'' contexts and
we prove many obstruction theorems, including Theorem A and Theorem
B. In Section \ref{sec_abstract_obstructions} EHP is internalized
into the ``abstract context'', where matrix algebras are replaced
by arbitrary (possibly graded) algebras. We then show how the ``fundamental
obstruction theorem'', namely Theorem A, naturally extends to this
context. We also give geometric obstructions to realize EHP into the
``full graded context'', i.e, into ``graded geometry'' and, therefore,
into ``supergeometry''. This section ends with a conjecture about
the existence of some kind of ``universal geometric obstruction''.
Finally, in Section \ref{sec_examples} many examples are given, where
by an ``example'' we mean an specific algebra fulfilling the hypothesis
of some obstruction result, so that EHP cannot be realized in the
underlying geometry.

\section{Polynomial Identities in Algebra-Valued Forms \label{sec_lie_alg}}

Let $A$ be a real vector space\footnote{All definitions and almost all results of this section hold analogously
for modules over any commutative ring $R$ of characteristic zero.
For some results in Subsection \ref{solv_algebras} we must assume
that the module is free, which is guaranteed if $R$ is a field.} and $P$ be a smooth manifold. A \emph{$A$-valued $k$-form in }$P$
is a section of the bundle $\Lambda^{k}TP^{*}\otimes A$. In other
words, it is a rule $\alpha$ assigning to any $a\in P$ a skew-symmetric
$k$-linear map 
\[
\alpha_{a}:TP_{a}\times...\times TP_{a}\rightarrow A.
\]
The collection of such maps inherits a canonical vector space structure
which we will denote by $\Lambda^{k}(P;A)$. Given an $A$-valued
$k$-form $\alpha$ and a $B$-valued $l$-form $\beta$ we can define
an $(A\otimes B)$-valued $(k+l)$-form $\alpha\otimes\beta$, such
that 
\[
(\alpha\otimes\beta)_{a}(v,w)=\frac{1}{(k+l)!}\sum_{\sigma}\operatorname{sign}(\sigma)\alpha_{a}(v_{\sigma})\otimes\beta_{a}(w_{\sigma}),
\]
where $v=(v_{1},...,v_{k})$ and $w=(v_{k+1},...,v_{k+l})$ and $\sigma$
is a permutation of $\{1,...,k+l\}$. This operator is obviously bilinear,
so that it extends to a bilinear map 
\[
\otimes:\Lambda^{k}(P;A)\times\Lambda^{l}(P;B)\rightarrow\Lambda^{k+l}(P;A\otimes B).
\]
We are specially interested when $A$ is an algebra, say with multiplication
$*:A\otimes A\rightarrow A$. In this case we can compose the operation
$\otimes$ above with $*$ in order to get an exterior product $\wedge_{*}$,
as shown below.\begin{equation}{\label{wedge_algebra}\xymatrix{\ar@/_{0.5cm}/[rr]_{\wedge_{*}}\Lambda^{k}(P;A)\times\Lambda^{l}(P;A)\ar[r]^{\otimes} & \Lambda^{k+l}(P;A\otimes A)\ar[r]^{*} & \Lambda^{k+l}(P;A)}}
\end{equation}Explicitly, following the same notations above, if $\alpha$ and $\beta$
are $A$-valued $k$ and $l$ forms, we get a new $A$-valued $(k+l)$-form
by defining 
\[
(\alpha\wedge_{*}\beta)_{a}(v,w)=\frac{1}{(k+l)!}\sum_{\sigma}\operatorname{sign}(\sigma)\alpha_{a}(v_{\sigma})*\beta_{a}(w_{\sigma}).
\]
This new product defines an $\mathbb{N}$-graded algebra structure
on the total $A$-valued space 
\[
\Lambda(P;A)=\bigoplus_{k}\Lambda^{k}(P;A),
\]
whose properties are deeply influenced by the properties of the initial
product $*$. For instance, recall that many properties of the algebra
$(A,*)$ can be characterized by its polynomial identities (PI's),
i.e, by polynomials that vanish identically when evaluated in $A$.
Just to mention a few, some of these properties are commutativity
and its variations (as skew-commutativity), associativity and its
variations (Jacobi-identity, alternativity, power-associativity),
and so on \cite{PI_1,PI_2}.

The following proposition shows that each such property is satisfied
in $A$ iff it is satisfied (in the graded-sense) in the corresponding
algebra of $A$-valued forms.
\begin{prop}
\label{PI} Any PI of degree $m$ in $A$ lifts to a PI in the graded
algebra of $A$-valued forms. Reciprocally, every PI in the graded
algebra restricts to a PI in $A$.
\end{prop}
\begin{proof}
Let $f$ be a polynomial of degree $m$. Given arbitrary $A$-valued
forms $\alpha_{1},...,\alpha_{m}$ of degrees $k_{1},...,k_{m}$,
we define a corresponding $A$-valued form of degree $k=k_{1}+...+k_{m}$
as 
\[
\hat{f}(\alpha_{1},...,\alpha_{m})(v_{1},...,v_{k})=\frac{1}{k!}\sum_{\sigma}\operatorname{sign}(\sigma)f(\alpha_{1}(w_{\sigma(1)}),\alpha_{2}(w_{\sigma(2)})...,\alpha_{m}(w_{\sigma(m)})),
\]
where 
\[
w_{\sigma(1)}=(v_{\sigma(1)},...,v_{\sigma(k_{1})}),\quad w_{2}=(v_{\sigma(k_{1}+1)},...,v_{\sigma(k_{1}+k_{2})}),\quad\text{and so on}.
\]
Therefore, $\hat{f}$ is a degree $m$ polynomial in $\Lambda(P;A)$.
It is clear that if $f$ vanishes identically then $\hat{f}$ vanishes
too, so that each PI in $A$ lifts to a PI in $\Lambda(P;A)$. On
the other hand, if we start with a polynomial $F$ in $\Lambda(P;A)$,
we get a polynomial $F\vert_{A}$ in $A$ of the same degree by restricting
$F$ to constant $1$-forms. Clearly, if $F$ is a PI, then $F\vert_{A}$
is also a PI. 
\end{proof}
It follows that the induced product $\wedge_{*}$ is associative,
alternative, commutative, and so on, iff the same properties are satisfied
by $*$. Notice that, because the algebra $\Lambda(P;A)$ is graded,
its induced properties must be understood in the graded sense. For
instance, by commutativity and skew-commutativity of $\wedge_{*}$
one means, respectively, 
\[
\alpha\wedge_{*}\beta=(-1)^{kl}\beta\wedge_{*}\alpha\quad\text{and}\quad\alpha\wedge_{*}\beta=-(-1)^{kl}\beta\wedge_{*}\alpha.
\]

Let us analyze some useful examples.
\begin{example}[\emph{Cayley-Dickson forms}]
\label{CD}There is a canonical construction, called \emph{Cayley-Dickson
construction} \cite{Cayley_Dickson}, which takes an algebra $(A,*)$
endowed with an involution $\overline{(-)}:A\rightarrow A$ and returns
another algebra $\operatorname{CD}(A)$. As a vector space it is just
$A\oplus A$, while the algebra multiplication is given by 
\[
(x,y)*(z,w)=(x*z-\overline{w}*y,z*x+y*\overline{z}).
\]
This new algebra inherits an involution $\overline{(x,y)}=(\overline{x},-y)$,
so that the construction can be iterated. It is useful to think of
$\operatorname{CD}(A)=A\oplus A$ as being composed of a ``real part''
and an ``imaginary part''. We have a sequence of inclusions into
the ``real part''$$
\xymatrix{A\ar[r] & \operatorname{CD}(A)\ar[r] & \operatorname{CD^{2}}(A)\ar[r] & \operatorname{CD^{3}}(A)\ar[r] & \cdots}
$$For any manifold $P$, such a sequence then induce a corresponding
sequence of inclusions into the algebra of forms$$
\xymatrix{\Lambda(P;A)\ar[r]&\Lambda(P;\operatorname{CD}(A))\ar[r]&\Lambda(P;\operatorname{CD^{2}}(A))\ar[r]&\Lambda(P;\operatorname{CD^{3}}(A))\ar[r]&\cdots}
$$

The Cayley-Dickson construction weakens any PI of the starting algebra
\cite{Cayley_Dickson,Cayley_Dickson_2} and, therefore, due to Proposition
\ref{PI}, of the corresponding algebra of forms. For instance, if
we start with the commutative and associative algebra $(\mathbb{R},\cdot)$,
endowed with the trivial involution, we see that $\operatorname{CD}(\mathbb{R})=\mathbb{C}$,
which remains associative and commutative. But, after an iteration
we obtain $\operatorname{CD}^{2}(\mathbb{R})=\operatorname{CD}(\mathbb{C})=\mathbb{H}$,
which is associative but not commutative. Another iteration gives
the octonions $\mathbb{O}$ which is non-assocative, but alternative.
The next is the sedenions $\mathbb{S}$ which is non-alternative.
\begin{example}[\emph{Lie algebra valued forms}]
\label{lie_algebra_forms}Another interesting situation occurs when
$(A,*)$ is a Lie algebra $(\mathfrak{g},[\cdot,\cdot])$\@. In this
case, we will write $\alpha[\wedge]_{\mathfrak{g}}\beta$ or simply
$\alpha[\wedge]\beta$ instead of\footnote{Some authors prefer the somewhat ambiguous notations $[\alpha\wedge\beta]$,
$[\alpha,\beta]$ or $[\alpha;\beta]$. The reader should be very
careful, because for some authors these notations are also used for
the product $[\wedge]_{\mathfrak{g}}$ without the normalizing factor
$(k+l)!$. Here we are adopting the notation introduced in \cite{marsh}.} $\alpha\wedge_{[\cdot,\cdot]}\beta$. Lie algebras are not associative
and in general are not commutative. This means that the corresponding
algebra $\Lambda(P;\mathfrak{g})$ is not associative and not commutative.
On the other hand, any Lie algebra is skew-commutative and satisfies
the Jacobi identity, so that these properties lift to $\Lambda(P;\mathfrak{g})$,
i.e, we have 
\[
\alpha[\wedge]\beta=(-1)^{kl+1}\beta[\wedge]\alpha
\]
and 
\[
(-1)^{km}\alpha[\wedge](\beta[\wedge]\gamma)+(-1)^{kl}\beta[\wedge](\gamma[\wedge]\alpha)+(-1)^{lm}\gamma[\wedge](\alpha[\wedge]\beta)=0
\]
for any three arbitrarily given $\mathfrak{g}$-valued differential
forms $\alpha,\beta,\gamma$ of respective degrees $k,l,m$.
\end{example}
\end{example}
We end with an important remark. 
\begin{rem}
\label{even}In the study of classical differential forms we know
that if $\alpha$ is an \textbf{odd}-degree form, then $\alpha\wedge\alpha=0$.
This follows directly from the fact that the algebra $(\mathbb{R},\cdot)$
is commutative. Indeed, in this case $\Lambda(P;\mathbb{R})$ is graded-commutative
and so, for $\alpha$ of \textbf{odd}-degree $k$, we have 
\[
\alpha\wedge\alpha=(-1)^{k^{2}}\alpha\wedge\alpha=-\alpha\wedge\alpha,
\]
implying the condition $\alpha\wedge\alpha=0$. Notice that the same
argument holds for any commutative algebra $(A,*)$. Dually, analogous
arguments show that if $(A,*)$ is skew-commutative, then $\alpha\wedge_{*}\alpha=0$
for any \textbf{even}-degree $A$-valued form. So, for instance, $\alpha[\wedge]_{\mathfrak{g}}\alpha=0$
for any given Lie algebra $\mathfrak{g}$. 
\end{rem}

\subsection{Matrix Algebras}

$\quad\;\,$In Example \ref{lie_algebra_forms} above, we considered
the algebra $\Lambda(P;\mathfrak{g})$ for an arbitrary Lie algebra
$\mathfrak{g}$. We saw that, because a Lie algebra is always skew-commutative,
the corresponding product $[\wedge]$ is also skew-commutative, but
now in the graded sense, i.e, 
\[
\alpha[\wedge]\beta=(-1)^{kl+1}\beta[\wedge]\alpha
\]
for every $\mathfrak{g}$-valued forms $\alpha,\beta$. As a consequence
(explored in Remark \ref{even}) we get $\alpha[\wedge]\alpha=0$
for even-degree forms.

From now on, let us assume that the Lie algebra $\mathfrak{g}$ is
not arbitrary, but a subalgebra of $\mathfrak{gl}(k;\mathbb{R})$,
for some $k$. In other words, we will work with Lie algebras of $k\times k$
real matrices. We note that in this situation $\Lambda(P;\mathfrak{g})$
can be endowed with an algebra structure other than $[\wedge]$. In
fact, for each $k$ there exists an isomorphism 
\[
\mu:\Lambda^{l}(P;\mathfrak{gl}(k;\mathbb{R}))\simeq\operatorname{Mat}_{k\times k}(P;\Lambda^{l}(P;\mathbb{R}))
\]
given by $[\mu(\alpha)]_{ij}(a)=[\alpha(a)]_{ij}$, allowing us to
think of every $\mathfrak{g}$-valued $l$-form $\alpha$ as a $k\times k$
matrix $\mu(\alpha)$ of classical forms. It happens that matrix multiplication
gives an algebra structure on 
\[
\operatorname{Mat}_{k\times k}(P;\Lambda(P;\mathbb{R}))=\bigoplus_{l}\operatorname{Mat}_{k\times k}(P;\Lambda^{l}(P;\mathbb{R})),
\]
which can be pulled-back by making use of the isomorphism $\mu$,
giving a new product on the graded vector space $\Lambda(P;\mathfrak{g})$,
which we will denote by the symbol ``$\curlywedge$''.

Because $\mathfrak{g}$ is a matrix Lie algebra, its Lie bracket is
the commutator of matrices, so that we have an identity relating both
products. From Proposition \ref{PI} it follows that we have an analogous
identity between the corresponding ``wedge products'' $\curlywedge$
and $[\wedge]$: 
\begin{equation}
\alpha[\wedge]\beta=\alpha\curlywedge\beta-(-1)^{kl}\beta\curlywedge\alpha.\label{relation_wedges_matrices}
\end{equation}
This relation clarifies that, while the product $[\wedge]$ is skew-commutative
for arbitrary Lie algebras, we \textbf{cannot} conclude the same for
the product $\curlywedge$, i.e, it is not always true that $\alpha\curlywedge\beta=(-1)^{kl+1}\beta\curlywedge\alpha$.
Consequently, \emph{it is }\textbf{\emph{not}}\emph{ true that $\alpha\curlywedge\alpha=0$
for every even-degree $\mathfrak{g}$-valued form}. It is easy to
understand why: recalling that $\curlywedge$ is induced by matrix
multiplication, the correspondence between PI's on the algebra $(\mathfrak{g},*)$
and on the algebra of $\mathfrak{g}$-forms teaches us that\emph{
$\curlywedge$ is skew-commutative exactly when the matrix multiplication
is skew-commutative}, in other words, iff $\mathfrak{g}$ is a Lie
algebra of skew-commutative matrices.

Useful examples are given in the following lemma. 
\begin{lem}
\label{sum_so} The condition $\alpha\curlywedge\alpha=0$ is satisfied
for even-degree forms with values into subalgebras 
\[
\mathfrak{g}\subset\mathfrak{so}(k_{1})\oplus...\oplus\mathfrak{so}(k_{r}).
\]
\end{lem}
\begin{proof}
By definition, $\mathfrak{so}(k)$ is the algebra of skew-symmetric
matrices, which anti-commute, thus the result holds for $\mathfrak{g}=\mathfrak{so}(k)$
for every $k$. It is obvious that if $\mathfrak{\overline{g}}$ is
some algebra of skew-commuting matrices, then every subalgebra $\mathfrak{g}\subset\mathfrak{\overline{g}}$
is also of skew-commuting matrices. Therefore, the result holds for
subalgebras $\mathfrak{g}\subset\mathfrak{so}(k)$. But it is also
clear that the direct sum of algebras of anti-commuting matrices remains
an algebra of anti-commuting matrices, so that if $\mathfrak{g}_{1}$
and $\mathfrak{g}_{2}$ fulfill the lemma, then $\mathfrak{g}_{1}\oplus\mathfrak{g}_{2}$
fulfills too. In particular, the lemma holds for subalgebras $\mathfrak{g}\subset\mathfrak{so}(k_{1})\oplus\mathfrak{so}(k_{2})$.
Finite induction ends the proof.
\end{proof}
\begin{rem}
\label{remark_sum_so}We could have given a proof of the last lemma
without using its invariance under direct sums. In fact, assuming
that it holds for subalgebras of $\mathfrak{so}(k)$, notice that
for every decomposition $k=k_{1}+...+k_{r}$ we have a canonical inclusion
\[
\mathfrak{so}(k_{1})\oplus...\oplus\mathfrak{so}(k_{r})\hookrightarrow\mathfrak{so}(k).
\]
\end{rem}
While the last lemma is a useful source of examples in which the discussion
above applies, let us now give a typical non-example.
\begin{example}
\label{example_mistake} Let $H$ be a linear group of $k\times k$
real matrices. We have a canonical action of $H$ on the additive
abelian group $\mathbb{R}^{k}$, allowing us to consider the semi-direct
product $G=\mathbb{R}^{k}\rtimes H$, of which $H\hookrightarrow G$
can be regarded as subgroup. The same holds in the level of Lie algebras,
so that $\mathfrak{g}\simeq\mathbb{R}^{k}\rtimes\mathfrak{h}$. As
a vector space we have $\mathfrak{g}/\mathfrak{h}\simeq\mathbb{R}^{k}$.
Now, because $\mathfrak{g}/\mathfrak{h}$ is abelian, it follows that
$\alpha[\wedge]\beta=0$ for every two $\mathfrak{g}/\mathfrak{h}$-valued
forms and from (\ref{relation_wedges_matrices}) we then conclude
that $\alpha\curlywedge\alpha=0$ for every odd-degree form $\alpha$.
In particular, if $A=\omega+e$ is a reductive Cartan connection for
the reduction $H\hookrightarrow\mathbb{R}^{k}\rtimes H$, then $e$
is a $\mathfrak{g}/\mathfrak{h}$-valued $1$-form, so that $e\curlywedge e=0$.
Consequently, the action (\ref{EHP_action}) is trivial. But, for
$H=O(3,1)$ this is precisely General Relativity, which is clearly
nontrivial. Therefore, at some moment we made a mistake! The problem
is that we applied the relation (\ref{relation_wedges_matrices})
to $\mathbb{R}^{n}$-valued forms, but the product $\curlywedge$
is \textbf{not} defined for such forms. Indeed, it can be defined
only for Lie algebras arising from matricial algebras. 
\end{example}
The example above teaches us two things:
\begin{enumerate}
\item The way it was defined, the product $\curlywedge$ makes sense only
for matrix Lie algebras. Particularly, it \textbf{does not} makes
sense for a semi-direct sum of a matrix algebra with other algebra,
so that $\curlywedge$ is \textbf{not} the product appearing in the
classical EHP action (\ref{EHP_action}). Therefore, if we need to
abstract the structures underlying EHP theory we need to show how
to extend $\curlywedge$ to a new product $\curlywedge_{\rtimes}$
defined on forms taking values in semi-direct sums. This will be done
in Subsection \ref{algebra_extensions};
\item Once $\curlywedge_{\rtimes}$ is defined, the corresponding Einstein-Hilbert-Palatini
action functional can be trivial. In fact, looking at (\ref{EHP_action-1})
we see that we have terms like $e\curlywedge_{\rtimes}e\curlywedge_{\rtimes}....\curlywedge_{\rtimes}e$,
where the number of $e$'s depends on the spacetime dimension. So,
if we are in an algebraic context in which $\alpha\curlywedge_{\rtimes}\alpha=0$
for \textbf{odd}-degree forms, then $e\curlywedge_{\rtimes}e=0$ and,
consequently, the action will vanishes in arbitrary spacetime dimensions!
On the other hand, if $\alpha\curlywedge_{\rtimes}\alpha=0$ for \textbf{even}-degree
forms, then $(e\curlywedge_{\rtimes}e)\curlywedge_{\rtimes}(e\curlywedge_{\rtimes}e)=0$
and the theory is trivial when $n\geq6$. \emph{This simple idea is
the core of almost all obstruction results that will be presented
here}. Notice that what we need is $\curlywedge_{\rtimes}^{k}\alpha=0$
for some $k$, which is a nilpotency condition. The correct nilpotency
conditions that will be used in our abstract obstruction theorems
will be discussed in Subsections \ref{nil_algebras} and \ref{solv_algebras}.
\end{enumerate}

\subsection{Splitting Extensions \label{algebra_extensions}}

$\quad\;\,$In this section we will see how to define the product
$\curlywedge_{\rtimes}$ abstractly. We start by recalling the notion
of ``extension of an algebraic object''. Let $\mathbf{Alg}$ be
a category of algebraic objects, meaning that we have a null object
$0\in\mathbf{Alg}$ (understood as the trivial algebraic entity) and
such that every morphism $f:X\rightarrow Y$ has a kernel and a cokernel,
computed as the pullback/pushout below. $$
\xymatrix{\ar[d]\operatorname{ker}(f)\ar[r] & 0\ar[d] &  & \operatorname{coker}(f) & \ar[l]0\\X\ar[r]_{f} & Y &  & Y\ar[u] & \ar[l]^{f}X\ar[u]}
$$
\begin{example}
The category $\mathbf{Grp}$ of groups and the category $\mathbf{Vec}_{\mathbb{R}}$
of real vector space are examples of models for $\mathbf{Alg}$. A
non-example is $\mathbf{LieAlg}$, the category of Lie algebras, since
a Lie algebra morphism $f:\mathfrak{h}\rightarrow\mathfrak{g}$ may
not have a cokernel, which means that the vector space $\mathfrak{g}/f(\mathfrak{h})$
may not have a canonical Lie algebra structure.
\end{example}
If a morphism has a kernel and a cokernel we can then define its \emph{image}
as the kernel of its cokernel. Consequently, internal to $\mathbf{Alg}$
we can talk of an ``exact sequence'': this is just an increasing
sequence of morphisms $f_{i}:X_{i}\rightarrow X_{i+1}$ such that
for every $i$ the kernel of $f_{i+1}$ coincides with the image of
$f_{i}$. An \emph{extension} is just a short exact sequence, i.e,
an exact sequence condensed in three consecutive objects. More precisely,
in a short exact sequence as the one below we say that the middle
term $A$ was obtained as an \emph{extension of $H$ by $E$}.\begin{equation}{\label{extension}
\xymatrix{0\ar[r] & E\ar[r]^{\imath} & A\ar[r]^{\jmath} & H\ar[r] & 0}}
\end{equation}
\begin{rem}
From the last example we conclude that we cannot talk of extensions
internal to the category of Lie algebras. This does \textbf{not} mean
that we cannot define a ``Lie algebra extension'' following some
other approach. In fact, notice that we have a forgetful functor $U:\mathbf{LieAlg}\rightarrow\mathbf{Vec}_{\mathbb{R}}$,
so that we can define a \emph{Lie algebra extension} as a sequence
of maps in $\mathbf{LieAlg}$ which is exact in $\mathbf{Vec}_{\mathbb{R}}$.
The same strategy allows us to enlarge the notion of extension in
order to include categories that, a priori, are not models for $\mathbf{Alg}$.
\end{rem}
Now, a typical example of extensions.
\begin{example}
The Poincaré group is just an extension of the Lorentz group by the
translational group. More generally, given a Lie group $H$ endowed
with an action $H\times\mathbb{R}^{n}\rightarrow\mathbb{R}^{n}$ we
can form the semi-direct product $\mathbb{R}^{n}\rtimes H$, which
fits into the canonical exact sequence below, where the first map
is an inclusion and the second is obtained restricting to pairs $(0,h)$.\begin{equation}{\label{canonical_extension}\xymatrix{0\ar[r] & \mathbb{R}^{n}\ar[r] & \mathbb{R}^{n}\rtimes H\ar[r] & H\ar[r] & 0}}
\end{equation}
\end{example}
The extensions in the last example are special: in them we know how
to include the initial object $H$ into its extension $\mathbb{R}^{n}\rtimes H$.
An extension with this property is called \emph{splitting}. More precisely,
we say that an abstract extension as (\ref{extension}) is \emph{splitting}
if there exists a morphism $s:H\rightarrow A$ such that $\jmath\circ s=id_{H}$.\begin{equation}{\label{split_extension}\xymatrix{0\ar[r] & E\ar[r]^{\imath} & A\ar[r]^{\jmath} & H\ar@/^{0.3cm}/[l]^{s}\ar[r] & 0}}
\end{equation}

The name comes from the fact that in some good situations, the category
$\mathbf{Alg}$ have a notion of ``product'', say $\#$, such that
a sequence (\ref{extension}) is splitting iff $A$ ``splits'' as
$A\simeq E\#H$. For instance, if $\mathbf{Alg}$ is an abelian category,
then the Splitting Lemma shows that such a product $\#$ is just the
coproduct $\oplus$, i.e, the ``direct sum''. For general groups
or Lie algebras, $\#$ is the corresponding notion of semi-direct
product/sum.

Now, assume that the ambient category $\mathbf{Alg}$ is actually
monoidal, meaning that we have a fixed bifunctor $\otimes:\mathbf{Alg}\times\mathbf{Alg}\rightarrow\mathbf{Alg}$
and an object $1\in\mathbf{Alg}$ such that $\otimes$ is associative
(up to natural isomorphisms) and has $1$ as a neutral object (also
up to natural isomorphisms). We can then talk of \emph{monoids in
$\mathbf{Alg}$}. These are objects $X\in\mathbf{Alg}$ endowed with
morphisms $*:X\otimes X\rightarrow X$ and $e:1\rightarrow X$ which
satisfy the associavitity-type and neutral element-type diagrams.

The main point is that we can use sections morphisms to transfer (to
pullback) monoid structures. Indeed, if $f:X\rightarrow Y$ is a morphism
with section $s:Y\rightarrow X$, then for any monoid structure $(*,e)$
in $Y$ we get a corresponding monoid structure $(*',e')$ in $Y$
with $*'$ as shown below and $e'=s\circ e$. In particular, if in
an splitting extension (\ref{split_extension}) $H$ is a monoid,
then we can use the section $s:H\rightarrow A$ to get a monoid structure
in $A$.\begin{equation}{\label{pullback_product}\xymatrix{\ar@/_{0.5cm}/[rrr]_{*'}X\otimes X\ar[r]^{f\otimes f} & Y\otimes Y\ar[r]^{*} & Y\ar[r]^s & X}}
\end{equation}

In the case when $\mathbf{Alg}$ admits a product $\#$ characterizing
splitting extensions it is natural to write $*_{\#}$ and $e_{\#}$
(instead of $*'$ and $e'$) to denote the monoid structure induced
on an extension.
\begin{example}
Let us take the category $\mathbf{Vec}_{\mathbb{R}}$ endowed with
the monoidal structure given by the tensor product $\otimes_{\mathbb{R}}$.
Its monoid objects are just real algebras. From the last paragraphs,
if $A$ is a splitting extension of a vector space $H$ by another
vector space $E$ and $H$ is an algebra, say with product $*$, then
$A$ is automatically an algebra with product $*_{\oplus}$. Therefore,
for a given manifold $P$ we will have not only a wedge product $\wedge_{*}$
between $H$-valued forms, but also a product $\wedge_{*_{\oplus}}$.
In the very particular case when $H$ is a matrix algebra, recall
that $\wedge_{*}$ is denoted by $\curlywedge$, which motivate us
to denote the corresponding $\wedge_{*_{\oplus}}$ by $\curlywedge_{\oplus}$.
It is exactly this kind of multiplication that appears in (\ref{EHP_action}).
\end{example}

\subsection{$(k,s)$-Nil Algebras \label{nil_algebras}}

$\quad\;\,$In the last subsection we showed how to build the products
that will enter in the abstract definition of EHP theory. Here we
will discuss the nilpotency conditions that will be imposed into these
products in order to get obstruction theorems for the corresponding
EHP theory.

We start by recalling that an element $v\in A$ in an algebra $(A,*)$
has \emph{nilpotency degree} $s$ if $v^{s}\neq0$, but $v^{s+1}=0$,
where $v^{i}=v*...*v$. In turn, the nilpotency degree of $A$ is
the minimum over the nilpotency degree of its elements. An algebra
with non-zero nilpotency degree is called a \emph{nil algebra}. This
``nil'' property can also be characterized as a PI: $A$ is nil
iff for some $s\neq0$ the polynomial $p_{1}(x)=x^{s}$ does not vanish
identically, but $x^{s+1}$ does.

The first nontrivial examples of nil algebras are the anti-commutative
algebras, which include Lie algebras, for which the nilpotency degree
is $s=1$. For such kind of objects we usually consider the more restrictive
notion of \emph{nilpotent algebra}. Indeed, we say that $A$ is \emph{nilpotent}
of degree $s$ if not only $p_{1}(x)=x^{s+1}$ vanishes, but also
\[
p_{s+1}(x_{1},...,x_{s+1})=x_{1}\cdot....\cdot(x_{s-1}\cdot(x_{s}\cdot x_{s+1}))
\]
and any other polynomial obtained from $p_{s+1}$ by changing the
ordering of the parenthesis. In general, being nilpotent is much stronger
than being nil. For associative algebras, on the other hand, such
concepts coincide \cite{nil_algebras_2,nil_algebras_1}. A fact more
easy to digest is that an anti-commutative and associative algebra
$A$ is nilpotent iff it is an \emph{Engel's algebra}, in the sense
that the PI's 
\[
p_{s,1}(x,y)=p_{s+1}(x,....,x,y)=0
\]
are satisfied for any ``$p_{s+1}$-type'' polynomial. Due to Engel's
theorem, other important examples of Engel's algebras are the\emph{
}Lie algebras.

From Proposition \ref{PI}, if an algebra $A$ is nil or nilpotent,
then the corresponding graded-algebra of $A$-valued exterior forms,
with the product $\wedge_{*}$, have the same properties, but in the
\textbf{graded} sense. Notice that $x^{s+1}=0$ iff $x^{s+1}=-x^{s+1}$,
so that the graded-nil condition becomes $x^{s+1}=(-1)^{k^{s+1}+1}x^{s+1}$,
where $k=\deg x$, which is nontrivial only when $k$ is even. The
conclusion is the following:\emph{ if $(A,*)$ is nil with nilpotency
degree $s>0$, then $\Lambda^{\operatorname{even}}(P;A)$ is graded-nil
with the same degree}.

The next example sets Lemma \ref{sum_so} in terms of this new language. 
\begin{example}
If $A$ is anti-commutative, then it is nil with degree $s=1$, which
implies that even $A$-valued forms are also nil with degree $s=1$,
i.e, $\alpha\wedge_{*}\alpha=0$ for even forms. If $A$ and $B$
are nil of respective degrees $r$ and $s$, then the direct sum $A\oplus B$
is also nil, with degree given by $\min\{r,s\}$. Subalgebras of nil
algebras are also nil algebras. Consequently, 
\[
A\hookrightarrow A_{1}\oplus...\oplus A_{r}
\]
is nil when each $A_{i}$ is nil. This is exactly a generalization
of Lemma \ref{sum_so}. 
\end{example}
The same kind of discussion applies to the nilpotent property. The
analysis is easier when $A$ is skew-commutative and satisfies an
associativity or Jacobi identity, because we can work with $p_{s,1}$
instead of $p_{s+1}$. Noting that $p_{s,1}$ is a PI iff $x^{s}\cdot y=-x^{s}\cdot y$
we conclude that the graded-nilpotent property is described by 
\[
x^{s}\cdot y=(-1)^{k^{s}l+1}(x^{s}\cdot y),
\]
where $k=\deg x$ and $l=\deg y$. This condition is nontrivial iff
$k^{s}l+1$ is odd, i.e, iff $k^{s}l$ is even, which implies that
$k$ and $l$ have the same parity. Summarizing: \emph{if $A$ is
nilpotent of degree $s$, then $\Lambda^{\operatorname{odd}}(P;A)$
or $\Lambda^{\operatorname{even}}(P;A)$ is graded-nilpotent with
the same degree}.

Notice that the nil property gives nontrivial conditions only on even-degree
forms. On the other hand, while the nilpotency property can be used
to give nontrivial conditions on even-degree or odd-degree forms,
it is too strong for most purposes. This leads us to define an intermediary
concept of \emph{$(k,s)$-nil algebra} containing both nil and nilpotent
algebras as particular examples. Indeed, given integers $k,s>0$,
we will say that an algebra $(A,*)$ is \emph{$(k,s)$-nil} if in
the corresponding graded-algebra $\Lambda(P;A)$ every $A$-valued
$k$-form has nilpotency degree $s$.

\subsection{$(k,s)$-Solv Algebras \label{solv_algebras}}

In this subsection we will explore more examples of $(k,s)$-nil algebras.
In the study of Lie algebras (as well as of other kinds of algebras),
there is a concept of \emph{solvable algebra }which is closely related
to the concept of \emph{nilpotent algebra}. Indeed, to any algebra
$A$ we can associate two decreasing sequences $A^{n}$ and $A^{(n)}$
of ideals, respectively called the \emph{lower central series }and
the \emph{derived series}, inductively defined as follows: 
\[
A^{0}=A,\quad A^{k}=\bigoplus_{i+j=k}A^{i}*A^{j}\quad\text{and}\quad A^{(0)}=A,\quad A^{(k)}=A^{(k-1)}*A^{(k-1)},
\]
where, for given subsets $X,Y\subset A$ , by $X*Y$ we mean the ideal
generated by all products $x*y$, with $x\in X$ and $y\in Y$. Clearly,
the polynomials like $p_{s+1}$ are PI's for $A$ iff $A^{s+1}=0$.
So, $A$ is nilpotent of degree $s$ iff its lower central series
stabilizes in zero after $s+1$ steps. Analogously, we say that $A$
is \emph{solvable of solvability degree $s$} if its derived series
stabilizes in zero after $s+1$ steps.

We notice that $A^{(k+1)}$ can be regarded not only as an ideal of
$A$, but indeed as an ideal of $A^{(k)}$. Furthermore, the quotient
$A^{(k)}/A^{(k+1)}$ subalgebra is always commutative \cite{nil_algebras_1}.
Starting with $A^{(1)}$ we get the first exact sequence (the first
line) below. Because we are working over fields, the quotient is a
free module and then the sequence splits, allowing us to write $A\simeq A^{(1)}\oplus A/A^{(1)}$.
Inductively we then get $A\simeq A_{s}'\oplus...\oplus A_{0}'$, where
$A_{i}':=A^{(i)}/A^{(i+1)}$. Summarizing: \emph{as a vector space,
a solvable algebra can be decomposed into a finite sum of spaces each
of them endowed with a structure of commutative algebra}.$$
\xymatrix{0\ar[r]&\ar[d]A^{(s)}\ar[r]&\ar[d]A^{(s-1)}\ar[r]&\ar[d]A/A^{(1)}\ar[r]&0\\&\vdots\ar[d]&\ar[d]\vdots&\ar[d]\vdots&\\0\ar[r]&\ar[d]A^{(2)}\ar[r]&\ar[d]A^{(1)}\ar[r]&\ar[d]A^{(1)}/A^{(2)}\ar[r]&0\\0\ar[r]&A^{(1)}\ar[r]&A\ar[r]&A/A^{(1)}\ar[r]&0}
$$

With the remarks above in our minds, let us prove that any solvable
algebra is ``almost'' a $(k,s)$-nil algebra. 
\begin{prop}
Let $(A,*)$ be a solvable algebra and let $\alpha$ be an $A$-valued
$k$-form in a smooth manifold $P$. If $k$ is odd and $\alpha$
is pointwise injective, then $\alpha\wedge_{*}\alpha=0$.
\end{prop}
\begin{proof}
Let $A\simeq A_{s}'\oplus...\oplus A_{0}'$ be the decomposition above.
Given a $k$-form $\alpha$, assume that $\alpha_{a}$ is injective
for every $a\in P$. Then, by the isomorphism theorem, $\alpha_{a}$
induces an isomorphism from its domain on its image. Notice that the
subespace $\operatorname{img}(\alpha)\subset A$ can be decomposed
as $\overline{A}_{s}\oplus...\overline{A}_{0}$, where $\overline{A}_{i}=\operatorname{img}(\alpha)\cap A_{i}'$.
Let $V_{i}$ be the preimage $\alpha_{a}^{-1}(\overline{A}_{i})$,
so that the domain of $\alpha_{a}$ decomposes as $V_{s}'\oplus...\oplus V_{1}'$,
allowing us to write $\alpha_{a}=\alpha_{a}^{s}+...+\alpha_{a}^{0}$.
We assert that $\alpha\wedge_{*}\alpha=0$. From the bilinearity of
$\wedge_{*}$ it is enough to verify that $\alpha_{a}^{i}\wedge_{*}\alpha_{a}^{j}=0$
for every $i,j$ and $a\in P$. If $i\neq j$ this is immediate because
$\alpha_{a}^{i}$ and $\alpha_{a}^{j}$ are nonzero in different subspaces.
So, let us assume $i=j$. In this case, since $A$ is solvable, the
algebras $A_{i}'$ are commutative, so that by the discussion in Section
\ref{sec_lie_alg} we have $\beta\wedge_{*}\beta=0$ for every $A_{i}'$-valued
odd-degree form and, in particular, $\alpha_{a}^{i}\wedge_{*}\alpha_{a}^{i}=0$. 
\end{proof}
Up to the injectivity hypothesis, the last proposition is telling
us that solvable algebras are $(k,1)$-nil algebras for every $k$
odd. So, we can say that solvable algebras are $(k,1)$-nil ``on
the class of injective forms''. This motivates us to define the following:
given a class $\mathcal{C}^{k}\subset\Lambda^{k}(P;A)$ of $A$-valued
$k$-forms, we say that $A$ is a \emph{$(\mathcal{C}^{k},s)$-nil
algebra} if any $A$-valued form belonging to $\mathcal{C}^{k}$ has
nilpotency degree $s$.

We are now in position of generalizing the last proposition. Indeed,
the structure of its proof is very instructive in the sense that it
can be easily abstracted by noticing that the ``solvable'' hypothesis
over $A$ was used only to get a decomposition of $A\simeq A_{s}'\oplus...\oplus A_{0}'$
in terms of commutative algebras. The commutativity, in turn, was
important only to conclude $\alpha\wedge_{*}\alpha=0$. Therefore,
using the same kind of proof we immediately obtain an analogous result
if we consider algebras $A$ endowed with a vector space decomposition
$A\simeq A_{s}\oplus...\oplus A_{0}$ where each $A_{i}$ is a $(k,s)$-nil
algebra. We will call this kind of algebras \emph{$(k,s)$-solv algebras},
because they generalize solvable algebras in the same sense as $(k,s)$-nil
algebras generalize nil and nilpotent algebras.

Summarizing, in this new terminology we have the following result. 
\begin{prop}
\label{solv_is_nil}Every $(k,s)$-solv algebra is \emph{$(\mathcal{C}^{k},s)$}-nil
over the class of pointwise injective forms.
\end{prop}
\begin{rem}
Exactly as nilpotent algebras are always solvable, $(k,s)$-nil algebras
are $(k,s)$-solv. The last proposition shows that the reciprocal
is almost true. On the other hand, we could have analogously defined
$(\mathcal{C}^{k},s)$-solv algebras and, in this case, if $\mathcal{C}^{k}$
is contained in the class of pointwise injective forms, the last proposition
would be rephrased as stating an equivalence between the concepts
of $(\mathcal{C}^{k},s)$-solv algebras and $(\mathcal{C}^{k},s)$-nil
algebras. 
\end{rem}
\begin{rem}
\label{solv_modules} In Subsection \ref{graded_valued_gravity} we
will work with special vector subspaces of a given algebra. Let us
seize the opportunity to introduce them. Given an arbitrary algebra
$A$, we say that a vector subspace $V\subset A$ is a \emph{$(k,s)$-nil
subspace of }$A$ if every $V$-valued $k$-form has nilpotency degree
$s$ when regarded as a $A$-valued form. Similarly, we say that $V$
is a \emph{$(k,s)$-solv} \emph{subspace }if it decomposes as a sum
of $(k,s)$-nil subspaces. When $A=\oplus_{m}A^{m}$ is $\mathfrak{m}$-graded,
there are other kind of subespaces $V\subset A$ that can be introduced.
For instance, we say that $V$ is \emph{graded $(k,s)$-solv} if each
$V_{m}=V\cap A_{m}$ is a $(k,s)$-solv subspace. In any subspace
$V$ of a $\mathfrak{m}$-graded algebra we get a corresponding grading
by $V\simeq\oplus_{m}V_{m}$. So, any $V$-valued form $\alpha$ can
be written as $\alpha=\sum_{m}\alpha^{m}$. We say that $V$ is \emph{weak
$(k,s)$-nilpotent} if for every $k$-form $\alpha$, any polynomial
$p_{s+1}(\alpha^{m_{1}},...,\alpha^{s+1})$ vanishes. Similarly, we
say that $V$ is \emph{weak $(k,s)$-solvable }if it decomposes as
a sum of \emph{weak $(k,s)$-nilpotent }subspaces. 
\end{rem}

\subsection{Functorial Algebra Bundle System\label{FABS}}

Now, let us discuss the last ingredient before applying Geometric
Obstruction Theory to EHP theories. In previous subsections, $P$
was an arbitrary smooth manifold. Let us now assume that it is the
total space of a $G$-bundle $\pi:P\rightarrow M$. EHP theories (which
are our aim) are not about forms on the total space $P$, but about
forms on the base manifold $M$. So, we need some process allowing
us to replace $A$-valued forms in $P$ by forms in $M$ with coefficients
in some other bundle, say $E_{A}$. More precisely, we are interested
in rules assigning to every pair $(P,A)$ a corresponding algebra
bundle $E_{A}$, whose typical fiber is $A$, in such a way that there
exists a graded-subalgebra $S(P,A)\subset\Lambda(P;A)$ and a canonical
morphism $\jmath:S(P,A)\rightarrow\Lambda(M;E_{A})$.

It is more convenient to think of this in categorical terms. Let $\mathbf{Alg}_{\mathbb{R}}$
be the category of real finite-dimensional\footnote{Actually, we can work in the infinite-dimensional setting, but in
this case we have to take many details into account. For instance,
we would need to work with topological algebras, bounded linear maps,
etc.} algebras, $\mathbb{Z}\mathbf{Alg}_{\mathbb{R}}$ be the category
of $\mathbb{Z}$-graded real algebras and, given a manifold $M$,
let $\mathbf{Bun}_{M}$ and $\mathbf{Alg_{\mathbb{R}}Bun}_{M}$ denote
the categories of bundles and of $\mathbb{R}$-algebra bundles over
$M$, respectively. As in the first diagram below, we have two canonical
functors, which assign to each pair $(A,P)$ the corresponding superalgebra
of $A$-valued forms in $P$, and to each algebra bundle $E$ over
$M$ the graded algebra of $E$-valued forms in $M$. We also have
the projection $(P,A)\mapsto A$. So, our problem of passing from
forms in $P$ to forms in $M$ could be solved by searching for a
functor $F$ making commutative the second diagram below.$$
\xymatrix{ & \mathbf{Alg_{\mathbb{R}}Bun}_{M}\ar[d]^{\Lambda(M;-)} & \mathbf{Alg_{\mathbb{R}}}\ar@{-->}[r]^{E_{-}} & \mathbf{Alg_{\mathbb{R}}Bun}_{M}\ar[d]^{\Lambda(M;-)}\\\mathbf{Bun}_{M}\times\mathbf{Alg_{\mathbb{R}}}\ar[r]_-{^{\Lambda(-;-)}} & \mathbb{Z}\mathbf{Alg_{\mathbb{R}}} & \ar[u]^{\pi_{1}}\mathbf{Bun}_{M}\times\mathbf{Alg_{\mathbb{R}}}\ar[r]_-{^{\Lambda(-;-)}} & \mathbb{Z}\mathbf{Alg_{\mathbb{R}}}}
$$

But this would be a stronger requirement; for instance, it would imply
the equality of two functors, a fact that can always be weakened by
making use of natural transformations. Therefore, we could search
for functors $E_{-}$ endowed with natural transformations $\jmath$,
as in the diagram below. This condition would remain stronger than
we need: it requires that for \textbf{any} bundle $P\rightarrow M$
and for \textbf{any} algebra $A$ we have a canonical algebra morphism
$\Lambda(P;A)\rightarrow\Lambda(M;E_{A})$. We would like to include
rules that are defined only for certain classes of bundles and algebras.
This leads us to work in subalgebras $\mathbf{C}$ of $\mathbf{Bun}_{M}\times\mathbf{Alg}_{\mathbb{R}}$
as in the second diagram below.$$
\xymatrix{\mathbf{Alg_{\mathbb{R}}}\ar@{-->}[r]^{E_{-}} & \mathbf{Alg_{\mathbb{R}}Bun}_{M}\ar[d]^{\Lambda(M;-)} & \mathbf{Alg_{\mathbb{R}}}\ar@{-->}[r]^{E_{-}} & \mathbf{Alg_{\mathbb{R}}Bun}_{M}\ar[d]^{\Lambda(M;-)}\\\ar[u]^{\pi_{1}}\ar@{==>}[ru]<-0.1cm>^{\jmath}\mathbf{Bun}_{M}\times\mathbf{Alg_{\mathbb{R}}}\ar[r]_-{^{\Lambda(-;-)}} & \mathbb{Z}\mathbf{Alg_{\mathbb{R}}} & \ar[u]^{\pi_{1}}\ar@{==>}[ur]+<-9pt>^-{\jmath}\mathbf{C}\ar[r]_-{^{\Lambda(-;-)}} & \mathbb{Z}\mathbf{Alg_{\mathbb{R}}}}
$$

But, once again, we should consider weaker conditions. Indeed, the
above situation requires that for $(P,A)\in\mathbf{C}$ the algebra
morphism $\Lambda(P;A)\rightarrow\Lambda(M;E_{A})$ is ``canonical''
in the entire algebra $\Lambda(P;A)$. It may happen that $\Lambda$
be ``canonical'' only in some subalgebra $S(P;A)\subset\Lambda(P;A)$,
meaning that it is first defined in $S(P;A)$ and then trivially extended
to $\Lambda(P;A)$. Therefore, the correct approach seems to be to
replace $\Lambda$ by another functor $S$ endowed with an objectwise
injective transformation $\imath:S\Rightarrow\Lambda$, as in the
following diagram.$$
\xymatrix{\mathbf{Alg_{\mathbb{R}}}\ar@{-->}[rr]^{E_{-}} &  & \mathbf{Alg_{\mathbb{R}}Bun}_{M}\ar[d]^{\Lambda(M;-)}\ar@{<==}[ld]+<18pt>_-{\jmath}\\\ar[u]^{\pi_{1}}\mathbf{C}\ar@{-->}@/^{0.4cm}/[rr]^<<<<<<<{S(-,-)} \ar@/_{0.4cm}/[rr]_-{^{\Lambda(-;-)}} & \ar@{==>}[]+<0.8cm,-0.5cm><0.3cm>_{\imath} & \mathbb{Z}\mathbf{Alg_{\mathbb{R}}}}
$$

In sum, given a manifold $M$,\emph{ the last diagram describes, in
categorical terms, the transition between algebra-valued forms in
bundles over $M$ and algebra bundle-valued forms in $M$}. The input
needed to do this transition corresponds to the dotted arrows in the
last diagram. We will say that they define a \emph{functorial algebra
bundle system }(FABS)\emph{ for the manifold $M$.} Concretely, a
FABS for $M$ consists of
\begin{enumerate}
\item a category $\mathbf{C}$ of pairs $(P,A)$, where $P\rightarrow M$
is a bundle and $A$ is a real algebra; 
\item a functor $E_{-}$ assigning to any algebra $A\in\mathbf{C}$ a corresponding
algebra bundle $E_{A}\rightarrow M$ whose typical fiber is $A$; 
\item a functor $S(-;-)$ that associates an algebra to each pair $(P,A)\in\mathbf{C}$;
\item natural transformations $\imath:S(-;-)\Rightarrow\Lambda(-;-)$ and
$\xi:S(-;-)\Rightarrow\Lambda(M;-)$ such that $\imath$ is objectwise
injective.
\end{enumerate}
$\quad\;\,$Such systems always exists, as showed by the next examples.
The fundamental properties and constructions involving FABS will appear
in a work under preparation \cite{FABS}.
\begin{example}[\emph{trivial case}]
Starting with any subcategory $\mathbf{C}$ of pairs $(P,A)$ that
contains only the trivial algebra, up to natural isomorphisms there
exists a single $E_{-}$: the constant functor at the trivial algebra
bundle $M\times0\rightarrow M$. Noting that $\Lambda(M;M\times0)\simeq0$
since the trivial algebra is a terminal object in $\mathbf{Alg}_{\mathbb{R}}$,
independently of the choice of $S$ and $\imath:S\Rightarrow\Lambda$
there is only one $\jmath:S\Rightarrow\Lambda(M;-)$: the trivial
one.
\end{example}
\begin{example}[\emph{almost trivial case}]
The last situation has the defect that it can be applied only for
subcategories $\mathbf{C}$ whose algebraic part is trivial. It is
easy, on the other hand, to build FABS for arbitrarily given algebras.
Indeed, let $\mathbf{C}$ be defined by pairs $(P,A)$, where $A$
is an arbitrary algebra, but $P$ is the trivial principal bundle
$M\times GL(A)\rightarrow M$. Putting $E_{A}$ as the trivial algebra
bundle $M\times A\rightarrow A$, we have a 1-1 correspondence between
$A$-valued forms in $P$ and forms in $M$ with values in $E_{A}$.
Therefore, we actually have a map $\jmath:\Lambda(P;A)\rightarrow\Lambda(M;E_{A})$
leading us to take $S=\Lambda$ and $\imath=\jmath$. 
\end{example}
\begin{example}[\emph{standard case}]
In the first example we considered $\mathbf{C}$ containing arbitrary
bundles, but we paid the price of working only with trivial algebras.
In the second example we were faced with a dual situation. A middle
term can be obtained by working with pairs $(P,A)$, where $P\rightarrow M$
is a $G$-bundle whose group $G$ becomes endowed with a representation
$\rho:G\rightarrow GL(A)$. In that case we define $E_{-}$ as the
rule assigning to each $A$ the corresponding associated bundle $P\times_{\rho}A$.
The functor $S$ is such that $S(P,A)$ is the algebra $\Lambda_{\rho}(P;A)$
of \emph{$\rho$-equivariant $A$-valued forms $\alpha$ in }$P$,
i.e, of those satisfy the equation $R_{g}^{*}\alpha=\rho(g^{-1})\cdot\alpha$,
where here $R:G\times P\rightarrow P$ is the canonical free action
characterizing $P$ as a principal $G$-bundle. This algebra of $\rho$-equivariant
forms naturally embeds into $\Lambda(P;A)$, giving $\imath$. Finally,
it is a standard fact \cite{kobayashi} that each $\rho$-equivariant
$A$-valued form on $P$ induces an $P\times_{\rho}A$-valued form
on $M$, defining the transformation $\xi$. This is the standard
approach used in the literature, so that we will refer to it as the
\emph{standard FABS}. 
\end{example}
\begin{example}[\emph{canonical case}]
There is an even more canonical situation: that obtained from the
latter by considering only associative algebras underlying matrix
Lie algebras. Indeed, in this case we have also a distinguished representation
$\rho:G\rightarrow GL(\mathfrak{g})$, given by the adjoint representation.
The corresponding bundles $E_{\mathfrak{g}}$ correspond to what is
known as the \emph{adjoint bundles}, leading us to say that this is
the \emph{adjoint FABS}.
\end{example}

\section{Concrete Obstructions \label{sec_obstruction}}

After the algebraic prolegomena developed in the previous section,
we are ready to introduce and study EHP theories in the general setting.
We will start in Subsection \ref{matrix_gravity} by considering a
very simple framework, modeled by inclusions of linear groups, where
our first fundamental obstruction theorem is obtained. At each following
subsection the theory will be redefined (in the direction of the full
abstract context) in such a way that the essence of the obstruction
theorem will remain the same.

Thus, in Subsection \ref{extended_linear_gravity} we define a version
of EHP for reductions $H\hookrightarrow G$ such that $G$ is not
a linear group, but a splitting extension of $H$. We then show that,
as desired, the fundamental obstruction theorem remains valid, after
minor modifications.

Independently, if one is interested in matrix EHP theories or EHP
theories arising from splitting extensions, we also work with reductive
Cartan connections. These are $1$-forms $\nabla:TP\rightarrow\mathfrak{g}$
that decompose as $\nabla=e+\omega$, where $e:TP\rightarrow\mathfrak{g}/\mathfrak{h}$
is a pointwise isomorphism and $\omega:TP\rightarrow\mathfrak{h}$
is an $H$-connection. What happens if we forget the pointwise isomorphism
hypothesis on $e$? In Subsection \ref{gauge_gravity} we discuss
the motivations to do this and we show that the fundamental obstruction
theorem not only holds, but becomes stronger.

Finally, in Subsection \ref{holonomy} we show that if we restrict
our attention to connections $\nabla=e+\omega$ such that $\omega$
is torsion-free, then we get, as a consequence of Berger's classification
theorem, a new obstruction result which is independent of the fundamental
obstruction theorem. In particular, this new result implies topological
obstructions to the spacetime in which the EHP theory is being described.

\subsection{Matrix Gravity \label{matrix_gravity}}

Let us start by considering a smooth $G$-bundle $P\rightarrow M$
over an orientable smooth manifold $M$, with $G$ a linear group,
and fix a group structure reduction $H\hookrightarrow G$ of $P$.
Given an integer $k>0$ and $\Lambda\in\mathbb{R}$, we define the\emph{
homogeneous} and the \emph{inhomogeneous linear} (or \emph{matrix})\emph{
Hilbert-Palatini form} of degree $k$ of a reductive Cartan connection
$\nabla=e+\omega$ in $P$, relative to the group structure reduction
$H\hookrightarrow G$, as 
\begin{equation}
\alpha_{k}=e\curlywedge...\curlywedge e\curlywedge\Omega\quad\text{and}\quad\alpha_{k,\Lambda}=\alpha_{k}+\frac{\Lambda}{(k-1)!}e\curlywedge...\curlywedge e,\label{hilbert_palatini_forms}
\end{equation}
respectively. In $\alpha_{k}$ the term $e$ appears ($k-2$)-times,
while in right-hand part of $\alpha_{k,\Lambda}$ it appears $k$-times.
Because we are working with matrix Lie algebras we have the adjoint
FABS, so that these $\mathfrak{g}$-valued forms correspond to $k$-forms
in $M$ with values in the adjoint bundle $P\times_{\operatorname{ad}}\mathfrak{g}$,
respectively denoted by $\jmath(\alpha_{k})$ and $\jmath(\alpha_{k,\Lambda})$.

We define the \emph{homogeneous }and the \emph{inhomogeneous} \emph{linear
EHP theories} in $P$ with respect to $H\hookrightarrow G$ as the
classical field theories whose spaces of configurations are the spaces
of reductive Cartan connections $\xi=e+\omega$ and whose action functionals
are respectively given by

\begin{equation}
S_{n}[e,\omega]=\int_{M}\operatorname{tr}(\jmath(\alpha_{n}))\quad\text{and}\quad S_{\Lambda,n}[e,\omega]=\int_{M}\operatorname{tr}(\jmath(\alpha_{n,\Lambda})).\label{EHP_action-2}
\end{equation}

\noindent \textbf{\uline{Warning:}}\textbf{ }The expression ``linear
EHP'' is used here to express the fact that the EHP action is realized
in a geometric background defined by linear groups. It \textbf{does
not} mean that the equations of motion were linearized or something
like this.
\begin{rem}
A priori, $H$ is an arbitrary subgroup of $G$, which implies that
$\mathfrak{h}$ is a Lie subalgebra of $\mathfrak{g}$. It would be
interesting to think of $\mathfrak{g}$ as the Lie algebra extension
of some $\mathfrak{r}$ by $\mathfrak{h}$. This makes sense only
when $\mathfrak{h}\subset\mathfrak{g}$ is not just a subalgebra,
but an ideal, which is true iff $H\subset G$ is a \textbf{normal}
subgroup.
\end{rem}
With the remark above in mind, we can state our first obstruction
theorems. In them we will consider theories for reductions $H\hookrightarrow G$
fulfilling one of the following conditions: 
\begin{enumerate}
\item[(C1)] the subgroup $H\subset G$ is normal and, as a matrix algebra, $\mathfrak{g}$
is an splitting extension by $\mathfrak{h}$ of a $(k,s)$-nil algebra
$\mathfrak{r}$. 
\item[(C2)] as a matrix algebra, $\mathfrak{g}$ is a $(k,s)$-nil algebra. 
\end{enumerate}
\begin{thm}
\label{theorem_C1}Let $M$ be an $n$-dimensional spacetime and $P\rightarrow M$
be a $G$-bundle, endowed with a group reduction $H\hookrightarrow G$
fulfilling (C1) or (C2). If $n\geq k+s+1$, then the linear inhomogeneous
EHP theory equals the homogeneous ones. If $n\geq k+s+3$, then both
are trivial.
\end{thm}
\begin{proof}
Assume (C1). As vector spaces we can write $\mathfrak{g}\simeq\mathfrak{g}/\mathfrak{h}\oplus\mathfrak{h}$
and, because $\mathfrak{g}$ is a Lie algebra extension of $\mathfrak{r}$
by $\mathfrak{h}$, it follows that $\mathfrak{g}/\mathfrak{h}\simeq\mathfrak{r}$.
How the extension splitting, $\mathfrak{r}$ is a subalgebra of $\mathfrak{g}$
and $\mathfrak{g}/\mathfrak{h}\simeq\mathfrak{r}$ is indeed a Lie
algebra isomorphism, So, $\mathfrak{g}/\mathfrak{h}$ can be regarded
as a $(k,s)$-nil algebra and, therefore, $\curlywedge^{s+1}\alpha=0$
for every $\mathfrak{g}/\mathfrak{h}$-valued $k$-form in $P$. If
$\nabla=e+\omega$ is a reductive Cartan connection in $P$ relative
to $H\hookrightarrow G$, then $e$ is a $\mathfrak{g}/\mathfrak{h}$-valued
$1$-form in $P$ and, because $\mathfrak{g}/\mathfrak{h}$ is a subalgebra
of $\mathfrak{g}$, $\curlywedge^{k}e$ is a $\mathfrak{g}/\mathfrak{h}$-valued
$k$-form. Consequently, 
\begin{equation}
(\curlywedge^{k+s+1}e)\curlywedge\alpha=0\label{obstruction_1}
\end{equation}
for any $\mathfrak{g}$-valued form $\alpha$. In particular, for
\[
\alpha=\frac{\Lambda}{(n-1)!}\curlywedge^{n-(k+s+1)}e
\]
(\ref{obstruction_1}) is precisely the inhomogeneous part of $\alpha_{\Lambda,n}$.
Therefore, for every $n\geq k+s+1$ we have $\alpha_{\Lambda,n}=\alpha_{n}$
and, consequently, $\jmath(\alpha_{k})=\jmath(\alpha_{\Lambda,n})$,
implying $S_{n}[e,\omega]=S_{n,\Lambda}[e,\omega]$ for any Cartan
connection $\nabla=e+\omega$, and then $S_{n}=S_{n,\Lambda}$. On
the other hand, for 
\[
\alpha=(\curlywedge^{n-(k+s+1)+2}e)\curlywedge\Omega
\]
we see that (\ref{obstruction_1}) becomes exactly $\alpha_{n}$.
Therefore, when $n\geq k+s+3$ we have $\alpha_{n}=0$, implying (because
$\jmath$ is linear) $\jmath(\alpha_{n})=0$ and then $S_{n}=0$.
But $k+s+3>k+s+1$, so that $S_{n,\Lambda}=0$ too. This ends the
proof when the reduction $H\hookrightarrow G$ fulfills condition
(C1). If we assume (C2) instead of (C1), we can follow exactly the
same arguments, but now thinking of $e$ as a $\mathfrak{g}$-valued
$1$-form.
\end{proof}
Recall that, as remarked in Section \ref{sec_lie_alg}, examples of
$(k,s)$-nil algebras $\mathfrak{g}$ include nil and nilpotent algebras.
In particular, from Lemma \ref{sum_so} we see that subalgebras of
$\mathfrak{so}(k_{1})\oplus...\oplus\mathfrak{s}(k_{r})$ are $(k,s)$-nil
for every $(k,s)$ such that $k$ is even and $s\geq1$. This leads
us to the following corollary:
\begin{cor}
\label{corollary_C1}For subalgebras as above, in a spacetime of dimension
$n\geq4$, the inhomogeneous and the homogeneous linear EHP are equal.
If $n\geq6$ both are trivial.
\end{cor}

\subsection{Gravity Arising From Extensions \label{extended_linear_gravity}}

Now, let us move on to a slightly more abstract situation: when $G$
is not a linear group, but a splitting extension of a given linear
group $H$, i.e, we will deal with group reductions $H\hookrightarrow\mathbb{R}^{k}\rtimes H$.
First of all, notice that if $P\rightarrow M$ is a bundle structured
over $H\subset GL(k;\mathbb{R})$, then we can always extend its group
structure to $\mathbb{R}^{k}\rtimes H$. Indeed, up to isomorphisms
this bundle is classified by a map $f:M\rightarrow BH$ and getting
a group extension is equivalent to lifting $f$ as shown below. It
happens that this lifting actually exists, since we are working with
splitting extensions (second diagram). It is in this context that
we will now internalize EHP theories.$$
\xymatrix{ & \mathbb{R}^{k}\rtimes H\ar[d]^{\jmath} &  &  & \mathbb{R}^{k}\rtimes H\ar[d]^{\jmath}\\M\ar@{-->}[ru]\ar[r]_{f} & H &  & \ar@{-->}[ru]^{s\circ f}M\ar[r]_{f} & H\ar@/_{0.5cm}/[u]_{s}}
$$

Thus, given a linear group $H$, consider bundles $P\rightarrow M$
endowed with group reductions $H\hookrightarrow\mathbb{R}^{k}\rtimes H$.
As previously, fixed $k>0$ and $\Lambda\in\mathbb{R}$ we define\emph{
homogeneous} and \emph{inhomogeneous extended-linear Hilbert-Palatini
form} of reductive Cartan connections $\nabla=e+\omega$ in $P$ as
\begin{equation}
\alpha_{k}=e\curlywedge_{\rtimes}...\curlywedge_{\rtimes}e\curlywedge_{\rtimes}\Omega\quad\text{and}\quad\alpha_{k,\Lambda}=\alpha_{k}+\frac{\Lambda}{(k-1)!}e\curlywedge_{\rtimes}...\curlywedge_{\rtimes}e.\label{hilbert_palatini_forms-1}
\end{equation}

Furthermore, fixed a FABS, the corresponding \emph{extended-linear
EHP theories} are defined by (\ref{EHP_action-2}). At first sight,
the only difference between the ``extended-linear'' and the ``linear''
theories is that we now consider the induced product $\curlywedge_{\rtimes}$
instead of $\curlywedge$. However, in the context of Geometric Obstruction
Theory, this replacement makes a fundamental difference. For instance,
the condition (C1) used to get Theorem \ref{theorem_C1} no longer
makes sense, because $H\hookrightarrow\mathbb{R}^{k}\rtimes H$ generally
is\textbf{ not }a normal subgroup. Immediate substitutes for (C1)
and (C2) are
\begin{enumerate}
\item[(C1')] as an associative algebra $(\mathfrak{h},\curlywedge_{\rtimes})$
is $(k,s)$-nil; 
\item[(C2')] as an associative algebra, $(\mathbb{R}^{k}\rtimes\mathfrak{h},\curlywedge_{\rtimes})$
is $(k,s)$-nil. 
\end{enumerate}
But, differently from (C1) and (C2), these new conditions are intrinsically
related. It is obvious that (C2)' implies (C1)', because $\mathfrak{h}$
is a subalgebra of $\mathbb{R}^{k}\rtimes\mathfrak{h}$. The reciprocal
is also valid, as shown in the next lemma.
\begin{lem}
\label{pullback_lemma}Let $(A,*)$ be an algebra, $B$ a vector space
and $f:A\rightarrow B$ a linear map that admits a section $s:B\rightarrow A$.
In this case, if $(A,*)$ is $(k,s)$-nil, then $(B,*')$ is too,
where $*'$ is the pulled-back multiplication (\ref{pullback_product}). 
\end{lem}
\begin{proof}
Recall that any algebra $X$ induces a corresponding graded-algebra
structure in $\Lambda(P;X)$ according to (\ref{wedge_algebra}).
Due to the functoriality of $\Lambda(P;-)$, we then get the commutative
diagram below, where the horizontal rows are just (\ref{wedge_algebra})
for the algebras $(A,*)$ and $(B,*')$, composed with the diagonal
map. The commutativity of this diagram says just that $\wedge_{*'}^{2}\alpha=\wedge_{*}^{2}f(\alpha)$
for every $\alpha\in\Lambda(P;B)$. From the same construction we
get, for each given $s\geq2$, a commutative diagram that implies
$\wedge_{*'}^{s}\alpha=\wedge_{*}^{s}f(\alpha)$. Therefore, if $(A,*)$
is $(k,s)$-nil it immediately follows that $(B,*')$ is $(k,s)$-nil
too.$$
\xymatrix{\ar[d]<0.2cm>^{\Lambda f}\Lambda(P;B)\ar[r]^-{\Delta} & \ar[d]<0.2cm>^{(\Lambda f)\otimes(\Lambda f)}\Lambda(P;B)\otimes\Lambda(P;B)\ar[r] & \ar[d]<0.2cm>^{\Lambda(f\otimes f)}\Lambda(P;B\otimes B)\ar[r] & \ar[d]<0.2cm>^{\Lambda f}\Lambda(P;B)\\\ar[u]<0.2cm>^{\Lambda s}\Lambda(P;A)\ar[r]^-{\Delta} & \Lambda(P;A)\otimes\Lambda(P;A)\ar[r]\ar[u]<0.2cm>^{(\Lambda s)\otimes(\Lambda s)} & \Lambda(P;A\otimes A)\ar[r]\ar[u]<0.2cm>^{\Lambda(s\otimes s)} & \ar[u]<0.2cm>^{\Lambda s}\Lambda(P;A)}
$$
\end{proof}
As a consequence of the lemma above, the version of Theorem (\ref{theorem_C1})
for extended-linear EHP is the following:
\begin{thm}
\label{theorem_C} Let $M$ be an $n$-dimensional spacetime and $P\rightarrow M$
be a $H$-bundle, where $H$ is a linear group such that $(\mathfrak{h},\curlywedge)$
is $(k,s)$-nil. If $n\geq k+s+1$, then the inhomogeneous extended-linear
EHP theory equals the homogeneous ones. If $n\geq k+s+3$, then both
are trivial.
\end{thm}
\begin{proof}
From the last lemma we can assume that $(\mathfrak{g},\curlywedge_{\rtimes})$
with $\mathfrak{g}=\mathbb{R}^{l}\rtimes\mathfrak{h}$ is $(k,s)$-nil.
An argument similar to that of the (C2)-case in Theorem \ref{theorem_C1}
gives the result. 
\end{proof}
\begin{cor}
\label{corollary_C} In a spacetime of dimension $n\geq4$, the cosmological
constant plays no role in any extended-linear EHP theory with $\mathfrak{h}\subset\mathfrak{so}(k_{1})\oplus...\oplus\mathfrak{so}(k_{r})$.
If $n\geq6$ the full theory is trivial. 
\end{cor}
\begin{rem}
This simple result will be our primary source of examples of geometric
obstructions. This is due to the fact that essentially all ``classical''
geometry satisfies the hypothesis of the last corollary, implying
that EHP cannot be realized in them in higher dimensions.
\end{rem}
\begin{rem}
\label{canonical_theory} For future reference, notice that any \emph{linear
EHP theory induces, in a canonical way, a corresponding extended linear
EHP theory}. Indeed, if $G\subset GL(k;\mathbb{R})$ is a linear Lie
group, then any subgroup $H\subset G$ inherits a canonical action
on $\mathbb{R}^{k}$, allowing us to consider the semidirect product
$\mathbb{R}^{k}\rtimes H$ and then $H$ as a subgroup of it. So,
if we start with a linear theory for the inclusion $H\hookrightarrow G$
we can always move to a a extended-linear theory for the inclusion
$H\hookrightarrow\mathbb{R}^{k}\rtimes H$. We will refer to this
induced theory as the \emph{canonical extended-linear theory associated}
to a given linear theory.
\end{rem}

\subsection{Gauged Gravity \label{gauge_gravity}}

In the last subsection we defined EHP theories in linear geometries
$H\hookrightarrow G$ and in extended-linear geometries $H\hookrightarrow\mathbb{R}^{k}\rtimes H$
as direct analogues of the EHP action functional. This means that
we considered theories on (reductive) Cartan connections for the given
group reduction, which are pairs $\nabla=e+\omega$, where $\omega$
is an usual $H$-connection and $e$ is a pointwise isomorphism.

In order to regard EHP theories as genuine (reductive) gauge theories,
it is necessary to work with arbitrary (i.e, not necessarily Cartan)
reductive connections. In practice, this can be obtained by just forgetting
the hypothesis of pointwise isomorphism in $e$. This leads us to
define \emph{homogeneous }and \emph{inhomogeneous gauge linear/extended-linear
Hilbert-Palatini forms }of a reductive (not necessarily Cartan) connection
$\nabla=e+\omega$ as in (\ref{hilbert_palatini_forms}). Similarly,
we can then define \emph{homogeneous }and \emph{inhomogeneous gauge
linear/extended-linear EHP theories} as in (\ref{EHP_action-2}),
whose configuration space is this new space of arbitrary reductive
connections.

A linear map may not be an isomorphism when it is not injective or/and
when it is not surjective. We would like that gauge EHP theories have
a nice physical interpretation, leading us to ask: \emph{what is the
physical motivation for $e$ be injective or/and surjective?}

At least for geometries described by $G$-structures (which in our
context means extend-linear geometries), the quotient $\mathfrak{g}/\mathfrak{h}$
can be identified with some $\mathbb{R}^{l}$ endowed with a tensor
$t$ such that the pair $(\mathbb{R}^{l},t)$ is the ``canonical
model'' for the underlying geometry. In such cases, the injectivity
hypothesis on $e_{a}:TP_{a}\rightarrow\mathfrak{g}/\mathfrak{h}$
is important to ensure that we have a good way to pull-back this ``canonical
geometric model'' to each fiber of $TP$. For instance, when $G=\mathbb{R}^{l}\rtimes O(n-1,1)$
and $H=O(n-1,1)$, the quotient $\mathfrak{g}/\mathfrak{h}$ is just
Minkowski space endowed with its standard metric $\eta$, and the
injectivity of $e$ implies that $g=e^{*}\eta$ is also a Lorentzian
metric $TP$. But, even if we forget injectivity, the tensor $g=e^{*}\eta$
\textbf{remains} well-defined, but now it is no longer non-degenerate.
Furthermore, for some models of Quantum Gravity, this non-degeneracy
is welcome (say to ensure the possibility of topology change \cite{quantum_gravity_1,quantum_gravity_2}),
which means that \emph{physically it is really interesting to consider
reductive connections $\xi=\omega+e$ such that $e$ is pointwise
non-injective}. With this in mind, we will do gauge EHP theories precisely
for this kind of reductive connections.

On the other hand, in the context of Geometric Obstruction Theory,
it is natural to expect that working with gauge EHP theories such
that $e$ in non-injective will impact the theory more than usual.
But, how deep will be this impact? Notice that the proofs of Theorem
\ref{theorem_C1} and Theorem \ref{theorem_C} were totally based
on the fact that in each point $e$ take values in some $(k,s)$-nil
algebra. In general, as discussed in Subsection \ref{solv_algebras},
being $(k,s)$-solv is weaker than being $(k,s)$-nil. But, if we
are in the class of injective forms, Proposition \ref{solv_is_nil}
shows that both concepts agree. Therefore, when working with Cartan
connections, no new results can be obtained if we replace ``$(k,s)$-nil''
with ``$(k,s)$-solv''. However, forgetting injectivity we may build
stronger obstructions.

In order to get these stronger obstructions, let us start by recalling
that until this moment we worked with the class of \textbf{reductive}
connections. These are $\mathfrak{g}$-valued $1$-forms $\xi$ which,
respectively to the vector space decomposition $\mathfrak{g}\simeq\mathfrak{g}/\mathfrak{h}\oplus\mathfrak{h}$,
can be globally written as $\xi=e+\omega$. This notion of ``reducibility''
extends naturally to $k$-forms with values in $A$, endowed with
some vector space decomposition $A\simeq A_{s}\oplus...\oplus A_{0}$.
Indeed, we say that an $A$-valued $k$-form $\alpha$ in a smooth
manifold $P$ is \emph{reductive }(or \emph{decomposable}) respectively
to the given vector space decomposition if it can be globally written
as $\alpha=\alpha_{s}+...+\alpha_{0}$. We can immediately see that
when $P$ is parallelizable every $A$-valued $k$-form is reductive
respectively to an arbitrarily given decomposition of $A$. It then
follows that \emph{in an arbitrary manifold algebra-valued $k$-forms
are locally reductive}.

After this digression, we state and proof the obstruction result for
gauge linear/extended-linear EHP theories. Here, once fixed a group
reduction $H\hookrightarrow G$, instead of (C1) and (C2) (which are
nilpotency conditions) we will consider the following analogous solvability
conditions:
\begin{enumerate}
\item[(S0)] the group $G$ is isomorphic to $\mathbb{R}^{l}\rtimes H$ and $(\mathfrak{h},\curlywedge)$
is $(k,s)$-solv; 
\item[(S1)] the subgroup $H\subset G$ is normal and, as a matrix algebra, $(\mathfrak{g},\curlywedge)$
is a splitting extension by $\mathfrak{h}$ of a $(k,s)$-solv algebra
$\mathfrak{r}$;
\item[(S2)] the full algebra $(\mathfrak{g},\curlywedge)$ is a $(k,s)$-solv
algebra.
\end{enumerate}
\begin{thm}
\label{theorem_S1}Let $M$ be an $n$-dimensional spacetime and $P\rightarrow M$
be a $G$-bundle, endowed with a group reduction $H\hookrightarrow G$
fulfilling \emph{(S0)}, \emph{(S1)} or \emph{(S2)}. If $n\geq k+s+1$,
then the inhomogeneous gauge EHP theory equals the homogeneous ones.
If $n\geq k+s+3$, then both are trivial.
\end{thm}
\begin{proof}
By definition, a $(k,s)$-solv algebra $A$ comes endowed with a decomposition
$A\simeq A_{s}\oplus...\oplus A_{0}$ by $(k,s)$-nil algebras $A_{i}$.
Therefore, under condition (S1) if, $e$ is an $\mathfrak{g}/\mathfrak{h}$-valued
$1$-form, where $\mathfrak{g}/\mathfrak{h}$ is $(k,s)$-solv, the
digression above tells us that it can be locally written as $e=e_{s}+...+e_{0}$.
Furthermore, we have 
\[
\curlywedge^{r}e=\curlywedge^{r}e_{s}+...+\curlywedge^{r}e_{0},
\]
because for $i\neq j$ the forms $e_{i}$ and $e_{j}$ are non-zero
in different spaces. Since each $e_{i}$ take values in a $(k,s)$-nil
algebra, it then follows that $\curlywedge^{k+s+1}e=0$ locally, which
implies $\curlywedge^{k+s+1}e\curlywedge\alpha=0$ for every $\mathfrak{g}$-valued
form $\alpha$. The remaining part of the proof is identical to that
of Theorem \ref{theorem_C1}, starting after equation (\ref{obstruction_1}).
Case (S2) is analogous to (S1) and (after using Lemma \ref{pullback_lemma})
case (S0) is analogous to (S2).
\end{proof}

\subsection{Dual Gravity \label{dual_gravity}}

$\quad\;\,$Here we will see that there are two ``dual theories''
associated to any EHP theory and we will see how the geometric obstructions
of the actual EHP theory relate to the obstructions affecting these
new theories.

We start by recalling that the fields in a gauge EHP theory (being
it linear or extended-linear) are reductive connections for a given
reduction $H\hookrightarrow G$. As a vector space we always have
$\mathfrak{g}\simeq\mathfrak{h}\oplus\mathfrak{g}/\mathfrak{h}$ and
these connections are given by $1$-forms $e:TP\rightarrow\mathfrak{g}/\mathfrak{h}$
and $\omega:TP\rightarrow\mathfrak{h}$ such that $\omega$ is a connection.
We can then think of a EHP theory as being determined by four variables:
the fields $e$ and $\omega$ together with the spaces $\mathfrak{g}/\mathfrak{h}$
and $\mathfrak{h}$. This allows us to define two types of \emph{``dual
EHP theories'' }by interchanging such variables.

More precisely, we define the \emph{geometric dual} of a given EHP
theory as that theory having the same space of configurations, but
whose action functional $^{*}S_{n,\Lambda}$ is $^{*}S_{n,\Lambda}[e,\omega]=S_{n,\Lambda}[\omega,e].$
Explicitly, for any fixed FABS we have 
\begin{equation}
^{*}S_{n,\Lambda}[e,\omega]=\int_{M}\operatorname{tr}(\jmath(^{*}\alpha_{n,\Lambda})),\quad\text{with}\quad{}^{*}\alpha_{n,\Lambda}=\curlywedge_{*}^{n-2}\omega\curlywedge_{*}E+\frac{\Lambda}{(n-1)!}\curlywedge_{*}^{n}\omega,\label{geometric_dual}
\end{equation}
where $E=de+e\curlywedge_{*}e$ and $\curlywedge_{*}$ must be interpreted
as $\curlywedge$ or $\curlywedge_{\rtimes}$ depending if the starting
theory is linear or extended-linear.

On the other hand, we define the \emph{algebraic dual} of a EHP theory
as that theory having an action functional with the same shape, but
now defined in a ``dual configuration space''. This is the theory
whose action function $S_{n,\Lambda}^{*}$ is given by 
\begin{equation}
S_{n,\Lambda}^{*}[e,\omega]=\int_{M}\operatorname{tr}(\jmath(\alpha_{n,\Lambda}^{*})),\quad\text{with}\quad\alpha_{n,\Lambda}^{*}=\curlywedge_{*}^{n-2}e\curlywedge_{*}\Omega+\frac{\Lambda}{(n-1)!}\curlywedge_{*}^{n}e,\label{algebraic_dual}
\end{equation}
where $e$ and $\omega$ take value in $\mathfrak{h}$ and $\mathfrak{g}/\mathfrak{h}$,
respectively. Furthermore, as above, $\curlywedge_{*}$ must be interpreted
as $\curlywedge$ or $\curlywedge_{\rtimes}$, depending of the case.

In very few words we can say that the \emph{geometric dual} of a given
theory makes a change at the \emph{dynamical level}, in the sense
that only the \emph{action funcional} is changed. Dually, the \emph{algebraic
dual} of a given theory produce changes only at the \emph{kinematic
level}, meaning that the dualization refers to the \emph{configuration
space}. This can be summarized in the table below. 
\begin{table}[H]
\begin{centering}
\begin{tabular}{|c|c|c|c|c|}
\hline 
$\operatorname{theories}/\text{variables}$  & $e$  & $\omega$  & $\mathfrak{g}/\mathfrak{h}$  & $\mathfrak{h}$\tabularnewline
\hline 
$\operatorname{EHP}$  & $e$  & $\omega$  & $\mathfrak{g}/\mathfrak{h}$  & $\mathfrak{h}$\tabularnewline
\hline 
$^{*}\operatorname{EHP}\;\,$  & $\omega$  & $e$  & $\mathfrak{g}/\mathfrak{h}$  & $\mathfrak{h}$\tabularnewline
\hline 
$\,\,\operatorname{EHP}^{*}$  & $e$  & $\omega$  & $\mathfrak{h}$  & $\mathfrak{g}/\mathfrak{h}$\tabularnewline
\hline 
\end{tabular}
\par\end{centering}
\caption{\label{table_dual}Relation between EHP theories and its geometric
and algebraic dualizations.}
\end{table}

Once such dual theories are introduced, let us now see how the obstruction
results for EHP theories apply to them. Let us start by focusing on
the geometric dual. Our main result until now is Theorem \ref{theorem_S1}.
It is essentially determined by two facts. First, $e$ take values
in an algebra fulfilling some ``solvability condition'' (i.e, any
one of hypotheses (S0), (S1) or (S2)); and second, the EHP action
contains powers of $e$.

When we look at the action (\ref{geometric_dual}) of the geometric
dual theory we see that it contains powers of $\omega$ (instead of
$e$). Therefore, assuming that $\omega$ take values in some ``solvable''
algebra we will get a result analogous to Theorem \ref{theorem_S1}.
Looking at Table \ref{table_dual} we identify that in the geometric
dual theory $\omega$ takes values in $\mathfrak{h}$ which is a subalgebra
of $\mathfrak{g}$, where the product $\curlywedge$ makes sense,
so that the ``solvability condition'' should be on $\mathfrak{h}$.
Dually, the action (\ref{algebraic_dual}) also contains powers, now
of $e$. Table \ref{table_dual} shows that here $e$ take values
in $\mathfrak{h}$, so that\emph{ if $(\mathfrak{h},\curlywedge)$
is ``solvable'', then not only the geometric dual theory is trivial,
but also the algebraic dual one!} Formally, we have the following
result, whose proof follows that of Theorem \ref{theorem_S1}
\begin{thm}
\label{theorem_duals}Let $M$ be an $n$-dimensional spacetime and
$P\rightarrow M$ be a $G$-bundle, endowed with a group reduction
$H\hookrightarrow G$. Assume that $(\mathfrak{h},\curlywedge)$ is
a $(k,s)$-solv algebra. If $n\geq k+s+3$, then both the geometric
dual and the algebraic dual of gauge EHP theory are trivial.
\end{thm}
\begin{cor}
\label{corollary_S0} If condition \emph{(S0)} is satisfied and $n\geq k+s+3$,
then gauge EHP theory and all its duals are trivial. 
\end{cor}
\begin{rem}
\label{no_S0_linear} In the \textbf{linear} context, where $G$ is
a linear Lie group and $H\subset G$ is a subgroup, Theorem \ref{theorem_duals}
above implies that if $(\mathfrak{h},\curlywedge)$ is $(k,s)$-solv,
then the duals of the actual EHP theory are trivial. This \textbf{does
not }mean that the actual EHP is trivial. In fact, Corollary \ref{corollary_S0}
needs condition (S0), which subsumes that we are working in the \textbf{extended-linear}
context.
\end{rem}

\subsection{The Role of Torsion \label{holonomy}}

$\quad\;\,$Until this moment we have given conditions under which
gauge EHP theories and their duals are trivial. One of this conditions
is (S0), which makes sense in the extended-linear context and states
that the subalgebra $(\mathfrak{h},\curlywedge)$ is $(k,s)$-solv.
Indeed, from Corollary \ref{corollary_S0}, if this condition holds
then EHP and their duals are all trivial for $n\geq k+s+3$. Due to
Lemma \ref{sum_so}, the standard examples of $(1,2)$-nil algebras
are subalgebras of $\mathfrak{so}(k_{1})\oplus...\oplus\mathfrak{so}(k_{r})$,
so that \emph{for such types of geometry EHP and their duals are trivial
if $n\geq6$}. Here we will see that if we work with a special class
of connections, then there are even more obstructions, due essentially
to Berger's classification theorem on Riemannian holonomy \cite{berger_original,berger_1}.

We start by recalling that if $\omega$ is a connection on a $G$-principal
bundle $\pi:P\rightarrow M$, each point $a\in P$ determines a Lie
subgroup $\operatorname{Hol}(\omega,a)\subset G$, called the \emph{holonomy
group} \emph{of $\omega$ in }$a$, which measures the failure of
a loop in $M$, based in $\pi(a)$, remaining a loop after horizontal
lifting \cite{kobayashi}. In typical situations, $P$ is the frame
bundle $FM$ of a $n$-manifold $M$, regarded as a $G$-structure
under some reduction $G\hookrightarrow GL(n;\mathbb{R})$. Such a
$G$-structure, in turn, generally is determined as the space of frame
transformations which preserve some additional tensor field $t$ in
$M$. In this context, the condition $\operatorname{Hol}(\omega,a)\subset G$
is equivalent to saying that there exists one such tensor field which
is parallel relative to $\omega$, i.e, such that $\nabla_{\omega}t=0$,
where $\nabla_{\omega}$ is the covariant derivative on tensors induced
by $\omega$. Particularly, $SO(n)$-reductions of $FM$ correspond
to Riemannian metrics on an oriented manifold, so that if $\omega$
is a connection on $FM$ such that $\operatorname{Hol}(\omega,a)\subset SO(n)$,
then $\nabla_{\omega}g=0$ for some Riemannian metric $g$. Consequently,
\emph{a torsion-free connection $\omega$ in $FM$ whose holonomy
is contained in $SO(n)$ is the Levi-Civita connection of some metric}.

Recall that, as pointed in Remark \ref{remark_sum_so}, for every
$k_{1}+...+k_{r}=n$ we have a canonical inclusion 
\begin{equation}
SO(k_{1})\times...\times SO(k_{r})\hookrightarrow SO(n)\label{de_Rham_reduction}
\end{equation}
as diagonal block matrices. Given a connection $\omega$ with special
Riemannian holonomy (meaning that it is contained in $SO(n)$) we
can ask: \emph{when is it indeed contained in the product subgroup
above?} From de Rham decomposition theorem \cite{de_Rham_original,de_Rham_thesis}
we see that if $\omega$ is torsion-free (and, therefore, the Levi-Civita
connection of a metric $g$) and $M$ is simply connected, this happens
iff $(M,g)$ is locally isometric\footnote{If $(M,g)$ is geodesically complete, the decomposition is global.}
to a product of Riemannian $k_{i}$-manifolds $(M_{i},g_{i})$, with
$i=1,...,r$, and $\operatorname{Hol}(\omega)$ is actually the product
of $\operatorname{Hol}(\omega_{i})\subset SO(k_{i})$, where $\omega_{i}$
is the Levi-Civita connection of $g_{i}$.

Assuming $(M,g)$ simply connected and locally irreducible in the
above sense, the holonomy reduction (\ref{de_Rham_reduction}) does
not exist. In this case, it is natural to ask for which proper subgroups
$G\hookrightarrow SO(n)$ the holonomy of the Levi-Civita connection
of $g$ can be reduced. When $(M,g)$ is not locally isometric to
a symmetric space, we have a complete classification of such proper
subgroups, given by \emph{Berger's classification theorem}, as in
Table \ref{berger}. 
\begin{table}[H]
\centering{}%
\begin{tabular}{|c|c|c|}
\hline 
$G\subset SO(n)$  & $\operatorname{dim}(M)$  & $\operatorname{nomenclature}$\tabularnewline
\hline 
\hline 
$U(n)$  & $2n$  & $\text{Kähler}$\tabularnewline
\hline 
$SU(n)$  & $2n$  & $\text{Calabi-Yau}$\tabularnewline
\hline 
$Sp(n)\cdot Sp(1)$  & $4n$  & $\text{Quaternionic-\text{Kähler}}$\tabularnewline
\hline 
$Sp(n)$  & $4n$  & $\text{Hyperkähler}$\tabularnewline
\hline 
$G_{2}$  & $7$  & $G_{2}$\tabularnewline
\hline 
$\operatorname{Spin}(7)$  & $8$  & $\operatorname{Spin}(7)$\tabularnewline
\hline 
\end{tabular}\caption{\label{berger}Berger's classification theorem for Riemannian signature.}
\end{table}

Let us see how we can use Berger's classification theorem to get geometric
obstructions similar to those given in Theorem \ref{theorem_duals}
and in Corollary \ref{corollary_S0}. We will need some definitions.
We say that a linear group $H$ is a \emph{$k$-group} if there are
nonnegative integers $k_{1},...,k_{r}$, with $k_{1}+...+k_{r}=k$,
such that $(\mathfrak{h},\curlywedge)$ is a subalgebra of $\mathfrak{so}(k_{1})\oplus...\oplus\mathfrak{so}(k_{r})$.
In other words, $k$-groups are fundamental examples of $(1,2)$-nil
algebras. Furthermore, we say that a manifold $N$ is a \emph{Berger
$k$-manifold }if it has dimension $k$ and it is simply connected,
locally irreducible and locally non-symmetric. A bundle $P\rightarrow M$
is called \emph{$k$-proper} if there exists an immersed Berger $k$-manifold
$N\hookrightarrow M$ whose frame bundle is a subbundle of $P$, i.e,
such that $FN\subset\imath^{*}P$. The motivating (trivial) examples
are the following:
\begin{example}
The frame bundle of a Berger $n$-manifold is, of course, $n$-proper. 
\begin{example}
If $f:N\hookrightarrow M$ an immersion with $N$ a Berger $k$-manifold,
then $FM$ is automatically $k$-proper, because $FN\subset f^{*}FM$.
\end{example}
\end{example}
Despite the notions introduced above, in the next theorem we will
work with \emph{torsionless }(\emph{extended-linear})\emph{ EHP theories},
i.e, extended-linear EHP theories restricted to reductive connections
$\nabla=e+\omega$ such that $\Theta_{\omega}=d_{\omega}e=0$. For
the case of linear EHP theories, recall that as discussed in Remark
\ref{canonical_theory}, to any of them we have an associated canonical
extended-linear theory.
\begin{thm}
\label{holonomy_EHP}Let $P\rightarrow M$ be a $k$-proper $H$-bundle
over a $n$-manifold $M$. If $H$ is a $k$-group, then $k_{1}=k$
and $k_{i>1}=0$. Furthermore, for any FABS, a torsionless extended-linear
EHP theory based on $H$ is nontrivial only if one of the following
conditions is satisfied
\begin{enumerate}
\item[\textit{\emph{(B1)}}] $k=2,4$ and $M$ contains a Kähler Berger $k$-manifold;
\item[\textit{\emph{(B2)}}] $k=4$ and $M$ contains a quaternionic-Kähler Berger $k$-manifold. 
\end{enumerate}
\end{thm}
\begin{proof}
From Lemma \ref{sum_so}, the hypothesis on $\mathfrak{h}$ implies
that it is $(1,2)$-nil and, therefore, because we are working with
extended-linear theories, condition (S0) is satisfied. Consequently,
by Theorem \ref{theorem_S1} the actual (and, in particular, the torsionless)
EHP theory is trivial if $n\geq6$, so that we may assume $n<6$.
Since $P\rightarrow M$ is $k$-proper, $M$ contains at least one
immersed Berger $k$-manifold $\imath:N\hookrightarrow M$ such that
$FN\subset\imath^{*}P$. Let $\kappa$ denote the inclusion of $FN$
into $\imath^{*}P$. On the other hand, we also have an immersion
$\imath^{*}P\hookrightarrow P$, which we denote by $\imath$ too.
Lie algebra-valued forms can be pulled-back and the pullback preserves
horizontability and equivariance. So, for any $\nabla=e+\omega$ the
corresponding $1$-form $(\imath\circ\kappa)^{*}\omega\equiv\omega\vert_{N}$
is an $H$-connection in $FN$ and its holonomy is contained in $H$.
By hypothesis $H$ is an $k$-group so that, via Lie integration,
\[
H\subset SO(k_{1})\times...\times SO(k_{r})\subset SO(k).
\]
In particular, the holonomy of $\omega\vert_{N}$ is contained in
$SO(k)$, implying that $\omega\vert_{N}$ is compatible with some
Riemannian metric $g$ in $N$. But, we are working with torsion-free
connections, so that $\omega\vert_{N}$ is actually the Levi-Civita
connection of $g$ and, because $N$ is irreducible, de Rham decomposition
theorem implies that there exists $i\in1,...,r$ such that $k_{i}=k$
and $k_{j}=0$ for $j\neq i$. Without loss of generality we can take
$i=1$. Because $N$ is a Berger manifold, Berger's theorem applies,
implying that the holonomy of $\omega\vert_{N}$ is classified by
Table \ref{berger}, giving conditions (B1) and (B2).
\end{proof}
\begin{cor}
\label{corollary_holonomy}Let $M$ be a Berger $n$-manifold with
an $H$-structure, where $H$ is a $n$-group. In this case, for any
FABS, a torsionless extended-linear EHP based on $H$ is nontrivial
only if $M$ has dimension $n=2,4$ and admits a Kähler structure.
\end{cor}
\begin{proof}
The result follows from the last theorem by considering the bundle
$P\rightarrow M$ as the frame bundle $FM\rightarrow M$ and from
the fact that every orientable four-dimensional smooth manifold admits
a quaternionic-Kähler structure \cite{einstein_manifolds_BESSE,quaternionic_Kahler_SALAMON}. 
\end{proof}
\begin{rem}
This corollary shows how topologically restrictive it is to internalize
\emph{torsionless} extended-linear EHP in geometries other than Lorentzian.
Indeed, if the spacetime $M$ is compact and 2-dimensional, then it
must be $\mathbb{S}^{2}$. On the other hand, in dimension $n=4$
compact Kähler structures exist iff the Betti numbers $b_{1}(M)$
and $b_{3}(M)$ are zero, so that $\chi(M)=b_{2}(M)+2$. As a consequence,
if we add the (mild) condition $H^{2}(M;\mathbb{R})\simeq0$ on the
hypothesis of Corollary \ref{corollary_holonomy} we conclude that
$M$ must be a K3-surface!
\end{rem}
The last theorem was obtained as a consequence of Theorem \ref{theorem_S1}
and of Berger's classification theorem. So, this is a geometric obstruction
result for \emph{extended-linear} EHP theories. However, due to Theorem
\ref{theorem_duals} and Corollary \ref{corollary_S0} we can get
exactly the same result for the geometric and algebraic \emph{duals}
of \emph{linear}\textbf{ }EHP. The same result does not hold for the
\emph{actual linear} EHP theories, because there is no analogue of
Theorem \ref{theorem_S1} or Corollary \ref{corollary_S0} for them,
as emphasized in Remark \ref{no_S0_linear}.

On the other hand, Berger's theorem remains valid, allowing us to
get an obstruction result for linear EHP theories independently of
the previous ones. Indeed, recall that (as pointed in Remark \ref{canonical_theory})
we can always associate a canonical (extended-linear) theory to a
given linear one. Let us define a \emph{torsionless liner EHP theory}
as that obtained by restricting a linear theory to the space of connections
such that the corresponding canonical theory is a torsionless extended-linear
EHP theory in the sense introduced above. We then have the following:
\begin{thm}
Let $P\rightarrow M$ be a $k$-proper $G$-bundle over a $n$-manifold
$M$ endowed with a group structure reduction $H\hookrightarrow G$.
If $H$ is a $k$-group, then $k_{1}=k$ and $k_{i>1}=0$. Furthermore,
a torsionless linear EHP theory for this reduction is nontrivial only
if one of the following conditions is satisfied
\begin{enumerate}
\item[\textit{\emph{(B1')}}] $k$ is even and $M$ contains a Kähler Berger $k$-manifold;
\item[\textit{\emph{(B2')}}] $k$ is divisible by $4$ and $M$ contains a quaternionic-Kähler
Berger $k$-manifold; 
\item[\textit{\emph{(B3')}}] $k=7$ and $M$ contains a Berger $k$-manifold with a $G_{2}$-structure.
\end{enumerate}
\end{thm}
\begin{proof}
The proof follows the same lines of the last theorem, except by the
first argument, which makes explicit use of Theorem \ref{theorem_S1}.
\end{proof}
The above results are about different versions of torsionless \emph{gauge}
EHP theories. By this we mean that no requirement was made on $e$.
We close this section remarking that if we work not on the \emph{gauge}
context, but on the \emph{Cartan} context (in the sense that $e$
is a pointwise isomorphism), then there is a physical appeal for working
with torsionless connections.

Indeed, recall that by the very abstract definition, a \emph{classical
theory }is given by an \emph{action functional} $S:\operatorname{Conf}\rightarrow\mathbb{R}$
defined on some \emph{space of configurations}. The interesting classical
theories are those in that $\operatorname{Conf}$ has some kind of
``smooth structure'' relative to which $S$ can be regarded as a
``smooth function'' and, as such, has a ``derivative''. In such
cases, there exists a distinguished subspace $\operatorname{Cut}(S)\subset\operatorname{Conf}$
constituted by the ``critical points'' of $S$. This is the \emph{phase
space} of the underlying classical theory, which contains all configurations
which a priori can be observed in nature.

If $S$ is now the action of EHP theory, we find that the phase space
is determined by the pairs $(e,\omega)$ which satisfy the equations
\[
\curlywedge_{*}^{n-2}e\curlywedge_{*}\Omega+\frac{\Lambda}{(n-1)!}\curlywedge_{*}^{n}e=0\quad\text{and}\quad\curlywedge_{*}^{n-2}e\curlywedge\Theta_{\omega}=0
\]
simultaneously, where (once again) $\curlywedge_{*}$ must be interpreted
as $\curlywedge$ or $\curlywedge_{\rtimes}$ depending if we are
in the linear or in the extended-linear context. The first of these
equations is just an analogue of Einstein's equations. The second,
in turn, if we are working with \emph{Cartan connections}, reduces
to $\Theta_{\omega}=0$, i.e, to the ``torsionless'' condition previously
used.

\section{Abstract Obstructions \label{sec_abstract_obstructions}}

Until this moment we considered EHP theory for group structure reductions
$H\hookrightarrow G$ other than $O(n-1,1)\hookrightarrow\operatorname{Iso}(n-1,1)$,
where $G$ could be an arbitrary linear group or some semidirect product
$\mathbb{R}^{k}\rtimes H$. In other words, we realized gravity, as
modeled by EHP theories, in other geometries than Lorentzian geometry.
We then showed that for certain classes of geometries the corresponding
theory is actually trivial, meaning that we have geometric obstructions.

In Subsection \ref{algebra_valued_gravity} we will generalize EHP
theories even more by replacing the algebras $\mathfrak{h}$ and $\mathfrak{g}$,
which are induced by the Lie groups $G$ and $H$, by general (not
necessarily satisfying PI's) algebras. We call these theories \emph{algebra-valued
EHP theories}\footnote{We could call such theories \emph{algebraic EHP theories}, but this
nomenclature would suggest that exactly as EHP theories are about
connections, ``algebraic EHP'' should be about ``algebraic connections''.
The problem is that the notion of ``algebraic connection'' actually
exists in the literature \cite{algebraic_connections_1,algebraic_connections_2,algebraic_connections_3},
being applied in a more general context. }. One motivation to consider this new generalization is the following.
In the previous extended-linear context, a group reduction $H\hookrightarrow\mathbb{R}^{n}\rtimes H$
on a manifold $M$ is (as discussed in the beginning of Subsection
\ref{holonomy}) a geometry modeled by some kind of tensor $t$ in
$M$. A \emph{global symmetry} of $(M,t)$ is an automorphism in the
category of $H$-structures, i.e, a difeo $ $ such that $f^{*}t=t$.
If we apply this to the canonical geometric model $(\mathbb{R}^{n},t)$
we see that its group of global symmetries $\operatorname{Aut}(\mathbb{R}^{n},t)$
is precisely $\mathbb{R}^{n}\rtimes H$. Therefore, the algebra $(\mathbb{R}^{n}\rtimes\mathfrak{h},\curlywedge_{\rtimes})$
takes the role of the associative algebra of \emph{infinitesimal global
symmetries of the canonical geometric model}, which is where the reductive
connections for the reduction $H\hookrightarrow\mathbb{R}^{n}\rtimes H$
take values.

It happens that in many physical situations we have more hidden/internal/worldsheet
symmetries, so that the \emph{full }algebra of infinitesimal symmetries
has a more abstract structure than just an associative algebra, leading
us to consider ``reductive connections with values in arbitrary algebras''.
Once the notion of algebra-valued EHP theoriesis introduced, we show
that the previous obstruction results (Theorem \ref{theorem_S1}),
hold \emph{ipsis litteris} in this new abstract context.

A typical situation in that the full algebra of infinitesimal symmetries
is not just an associative algebra is when the underlying ``canonical
geometric model'' is not a cartesian space $\mathbb{R}^{n}$ endowed
with a tensor, but actually a supercartesian space $\mathbb{R}^{n\vert m}$
endowed with a supertensor. Indeed, in this case the algebra of infinitesimal
symmetries inherits a $\mathbb{Z}_{2}$-graded algebra structure.
This motivates us to analyze the effects of gradings in the obstruction
results, which is done in Subsection \ref{graded_valued_gravity}.

Notice that if the pair $(\mathbb{R}^{n},t_{n})$ describes geometry,
the pair $(\mathbb{R}^{n\vert m},t_{n\vert m})$ describes \emph{supergeometry}.
Therefore, it would be natural to consider not only ``EHP with values
in superalgebras'', but actually ``super EHP theories''. Closing
the section, in Subsection \ref{graded_gravity} we present an approach
to the notion of ``reductive graded-connection'', allowing us to
define ``graded EHP theories'', and we show how to extend the obstruction
results to this context.

\subsection{Algebra-Valued Gravity \label{algebra_valued_gravity}}

We start by recalling that a reductive connection in $P$ for $H\hookrightarrow G$
is a $\mathfrak{g}$-valued 1-form $\nabla$ which decomposes as $\omega+e$,
where $\omega$ is a $\mathfrak{h}$-valued $1$-form and $e$ is
a $\mathfrak{g}/\mathfrak{h}$-valued $1$-form (not necessarily an
isomorphism, due to previous discussion). Notice that it is completely
determined by the pair $(\omega,e)$ and by the vector space decomposition
$\mathfrak{g}\simeq\mathfrak{h}\oplus\mathfrak{g}/\mathfrak{h}$,
a fact that was already used in Subsection \ref{dual_gravity}. In
typical cases $\mathfrak{g}$ is actually a splitting extension (the
previous extended-linear context) or $H\subset G$ is a normal subgroup
(previous linear context), so that both $\mathfrak{h}$ and $\mathfrak{g}/\mathfrak{h}$
acquire algebra structures.

This leads us to the following generalization: given an $\mathbb{R}$-algebra
$(A,*)$ endowed with a vector space decomposition $A\simeq A_{0}\oplus A_{1}$,
where $A_{0}$ and $A_{1}$ have algebra structures $*_{0}$ and $*_{1}$
(not necessarily subalgebras of $A$), we define an\emph{ $A$-connection}
in a manifold $P$ as an $A$-valued $1$-form $\nabla$ in $P$ which
is reductive respective to $A\simeq A_{1}\oplus A_{0}$. In \}

\{kmore concrete terms, it is an $A$-valued $1$-form $\nabla$ which
is written as $\nabla=\omega+e$, where $\omega$ and $e$ take values
in $A_{1}$ and $A_{0}$, respectively. The \emph{curvature} of $\omega$
is the $2$-form $\Omega=d\omega+\omega\wedge_{1}\omega$ in $A_{1}$,
where $\wedge_{1}$ is the product induced by $*_{1}$. Similarly,
the \emph{torsion }of $\omega$ is the $A$-valued $2$-form $\Theta_{\omega}=de+\omega\wedge_{*}e$.
Some related concepts are considered in \cite{related_article_1,related_article_2}.

With these structures in hand we can generalize gauge EHP theories
(about $\mathfrak{g}$-valued connections) to \emph{algebra-valued
EHP theories} (about $A$-connections in the above sense). Indeed,
given a spacetime manifold $M$, this generalization is obtained following
the following steps:
\begin{enumerate}
\item consider some analogue of the Hilbert-Palatini forms (\ref{hilbert_palatini_forms}); 
\item show that this Hilbert-Palatini form induces a corresponding form
in $M$ with values in some bundle; 
\item turn this bundle-valued form into a real-valued form;
\item define the action functional as the integral over $M$ of this real-valued
form.
\end{enumerate}
In order to do the second step, the immediate idea is to select a
FABS, say defined on a subcategory $\mathbf{C}\subset\mathbf{Bun}_{M}\times\mathbf{Alg}_{\mathbb{R}}$,
as in Section \ref{FABS}. However, in order to realize the first
step we need to work with FABS fulfilling certain ``invariance property'':
we say that a FABS on $\mathbf{C}$ is \emph{invariant }by a functor
$I:\mathbf{C}\rightarrow\mathbb{Z}\mathbf{Alg}_{\mathbb{R}}$ if
\begin{enumerate}
\item[(a)] for all $(P,A)\in\mathbf{C}$ the corresponding $I(P,A)$ is an ideal
of $S(P;A)$, so that we can take the quotient functor $S/I$ and
we have a natural transformation $\pi:S\Rightarrow S/I$;
\item[(b)] there exists another functor $J:\mathbf{C}\rightarrow\mathbb{Z}\mathbf{Alg}_{\mathbb{R}}$
such that $J(P,A)$ is an ideal of $\Lambda(M;E_{A})$ and whose projection
we denote by $\pi'$; 
\item[(c)] there exists a natural transformation $\jmath':S/I\Rightarrow\Lambda(M;-)/J$
such that $\jmath'\circ\pi=\pi'\circ\jmath$, i.e, the diagram below
commutes for every $(P,A)$.\begin{equation}{\label{invariant_FABS}\xymatrix{\ar[d]_{\pi_{(P,A)}}S(P,A)\ar[r]^{\jmath_{(P,A)}} & \Lambda(M;E_{A})\ar[d]^{\pi'_{(P,A)}}\\S(P,A)/I(P,A)\ar[r]_{\jmath'_{(P,A)}} & \Lambda(M;E_{A})/J(P,A)}}
\end{equation}
\end{enumerate}
Returning to the first step, notice that a priori the algebras $A$
and $A_{1}$ may not be associative, so that due to ambiguities we
cannot simply replace ``$\curlywedge$'' by ``$\wedge_{*}$''
and ``$\wedge_{0}$'' in (\ref{hilbert_palatini_forms}) to get
an $A$-valued version of $\alpha_{HP}$; instead, we should take
into account all possibilities simultaneously. This can be done as
follows. Recall the ``associativity-like'' polynomials 
\[
(x_{1}\cdot...\cdot(x_{s-2}\cdot(x_{s-1}\cdot x_{s}))),\quad(((x_{1}\cdot x_{2})\cdot(x_{3}\cdot x_{4}))\cdot...\cdot x_{s}),\quad\text{etc}.,
\]
which were introduced in Subsection \ref{nil_algebras} in order to
describe nilpotency and nil properties as PI's. Due to the structure
of the Hilbert-Palatini forms, we are interested in the more specific
polynomials 
\begin{equation}
(x\cdot...\cdot(x\cdot(x\cdot y))),\quad\text{etc.}\quad\text{and}\quad(x\cdot...\cdot(x\cdot(x\cdot x))),\quad\text{etc}.\label{p_s}
\end{equation}
When evaluated at $\Lambda(P;A)$, the set of these polynomials generates
an ideal $\mathfrak{I}(P,A)$ and the idea is to consider FABS which
are invariant by it.

More precisely, given a bundle $P\rightarrow M$ over a manifold $M$
and an algebra $A$ endowed with vector space decomposition $A\simeq A_{0}\oplus A_{1}$,
we say that a FABS over $M$, defined in $\mathbf{C}$, is \emph{compatibl}e
with the given data if 
\begin{enumerate}
\item[(a)] the pair $(P,A)$ belongs to $\mathbf{C}$;
\item[(b)] each possibility of defining $\alpha_{n,\Lambda}$ (i.e, the evaluation
of each polynomial (\ref{p_s}) in each $A$-connection) belongs to
$S(P,A)$;
\item[(c)] the FABS in question is invariant by $\mathfrak{I}$, as defined
above.
\end{enumerate}
Considering FABS satisfying (a)-(c) we realize the first two steps
needed to define ``algebra-valued EHP theory''. In order to realize
the third step, we need something like a ``trace''. A \emph{trace
transformation} for a FABS on $\mathbf{C}$ invariant by $I$ is a
natural transformation $\operatorname{tr}$ between $\Lambda(M;E_{-})/J$
and the constant functor in $\Lambda(M;\mathbb{R})$. In other words,
it is a rule that assigns (in a natural way) a map of graded-algebras
$\operatorname{tr}_{(P,A)}:\Lambda(M;E_{A})\rightarrow\Lambda(M;\mathbb{R})$
to every pair $(P,A)\in\mathbf{C}$.

Now, we can finally define what is an EHP theory in the general algebraic
setting. Given a principal bundle $P\rightarrow M$, an algebra $A$
and a vector space decomposition $A\simeq A_{0}\oplus A_{1}$, choose
a compatible FABS endowed with a trace transformation. The corresponding
\emph{$A$-valued} \emph{inhomogeneous}\footnote{The corresponding $A$-valued \emph{homogeneous }EHP theory is defined
analogously.}\emph{ EHP theory} in $P$ is the classical theory whose configuration
space is the space of all reductive $A$-connections in $P$ and whose
action functional is given by (we omit the subscripts in the maps
$\pi$, $j'$ and $\operatorname{tr}$ in order to simplify the notation)
\[
S_{n,\Lambda}[e,\omega]=\int_{M}\operatorname{tr}(\jmath'(\pi(\alpha_{HP}))),
\]
where $\alpha_{HP}\in S(P,A)$ is any ``Hilbert-Palatini''-type
form, say 
\[
\alpha_{HP}=(e\wedge_{*}...(\wedge_{*}(e\wedge_{*}(e\wedge_{*}\Omega)))+\frac{\Lambda}{(n-1)!}(e\wedge_{*}...(\wedge_{*}(e\wedge_{*}(e\wedge_{*}\Omega))).
\]

Notice that we recover the gauge linear EHP theories discussed previously
by taking $A$ as a matrix Lie algebra $\mathfrak{g}$ and considering
the adjoint FABS endowed with the classical trace map. Therefore,
at least in this restricted domain we have obstruction Theorems \ref{theorem_C1},
\ref{theorem_S1} and \ref{theorem_duals}. These results remain valid
if: 
\begin{enumerate}
\item \emph{we replace the adjoint FABS by any other compatible FABS} (indeed,
due to the linearity of the FABS, once one shows that $\alpha_{HP}=0$
it immediately follows that the bundle-valued form $\jmath'(\pi(\alpha_{HP}))$
is also trivial independently of the FABS chosen); 
\item \emph{we replace the classical trace by any other trace transformation}
(also due to the linearity of trace transformations); 
\item \emph{we replace the algebra $\mathfrak{g}$ by any other algebra
$A$ endowed with a vector space decomposition $A\simeq A_{0}\oplus A_{1}$}
(this follows from the structure of the proofs of Theorems\ref{theorem_C1},
\ref{theorem_S1} and \ref{theorem_duals}, which in essence depends
only on general algebraic hypotheses on the decomposition $\mathfrak{g}\simeq\mathfrak{g}/\mathfrak{h}\oplus\mathfrak{h}$).
\end{enumerate}
Summarizing: the obstruction theorems hold not only in the domain
of matrix algebras, adjoint FABS and classical trace, but also for
arbitrary algebras, arbitrary compatible FABS and arbitrary trace
transformations. Concretely, we have the following general obstruction
theorem whose proof is in essence Remarks 1-3 above.
\begin{thm}
\label{theorem_A1}Let $M$ be an $n$-dimensional spacetime and $P\rightarrow M$
be a bundle and $(A,*)$ an algebra endowed with a vector space decomposition
$A_{0}\oplus A_{1}$ fulfilling one of the following conditions
\begin{enumerate}
\item[\textit{\emph{(}}A1\textit{\emph{)}}] the algebra $(A_{0},*_{0})$ is a $(k,s)$-solv subalgebra\footnote{We can forget the subalgebra hypothesis by modifying a little bit
the ideal $\mathfrak{I}(P,A)$ used to define the notion of ``compatible
FABS''.} $(A,*)$;
\item[\textit{\emph{(}}A2\textit{\emph{)}}] $(A,*)$ is itself $(k,s)$-solv.
\end{enumerate}
If $n\geq k+s+1$, then for any compatible FABS and any trace transformation,
the corresponding inhomogeneous $A$-valued EHP theory equals the
homogeneous ones. If $n\geq k+s+3$, then both theories are trivial.
\end{thm}

\subsection{Graded-Valued Gravity \label{graded_valued_gravity}}

Here we shall indicate how the previous discussion can be extended
to the case when the background algebra $A$ is itself graded (for
more details see \cite{FABS}). Given a monoid $\mathfrak{m}$, let
$P\rightarrow M$ be a principal bundle and $A$ be a $\mathfrak{m}$-graded
real algebra. For any vector space decomposition $A\simeq A_{0}\oplus A_{1}$
where $A_{i}$ are algebras, we define a \emph{connection} \emph{in
$P$ with values in the graded algebra $A$ }exactly as previously:
as a $A$-valued 1-form $\nabla$ in $P$ which writes as $\nabla=e+\omega$
for $e:TP\rightarrow A_{0}$ and $\omega:TP\rightarrow A_{1}$. The
only difference is that now the $\mathfrak{m}$-grading of $A$ induces
a corresponding grading in each $A_{i}$ by $\oplus_{m}A_{i}^{m}$,
where $A_{i}^{m}=A_{i}\cap A^{m}$, which means that locally we can
write $e$ and $\omega$ as 
\[
e=\sum{}_{m}e^{m}\quad\text{and}\quad\omega=\sum_{m}\omega^{m},
\]
where $e^{m}:TP\rightarrow A_{0}^{m}$ and $\omega:TP\rightarrow A_{1}^{m}$
are just the projections of $e$ and $\omega$ onto the corresponding
$A_{i}^{m}$.

In order to extend FABS to this graded context we notice that if $A$
is $\mathfrak{m}$-graded, then the algebra of $A$-valued exterior
forms $\Lambda(P;A)$ is ($\mathbb{Z}\times\mathfrak{m}$)-graded.
On the other hand, an algebra bundle $E_{A}$ whose typical fiber
is $A$ is not necessarily $\mathfrak{m}$-graded, because the pointwise
decomposition may not vary continuously to allow us to globally decompose
$E_{A}$ as $\oplus_{m}E_{A_{m}}$. Therefore, in order to incorporate
graded-algebras it is enough to work with ``graded FABS'' characterized
by the following diagram: 
\begin{center}
\includegraphics[scale=0.35]{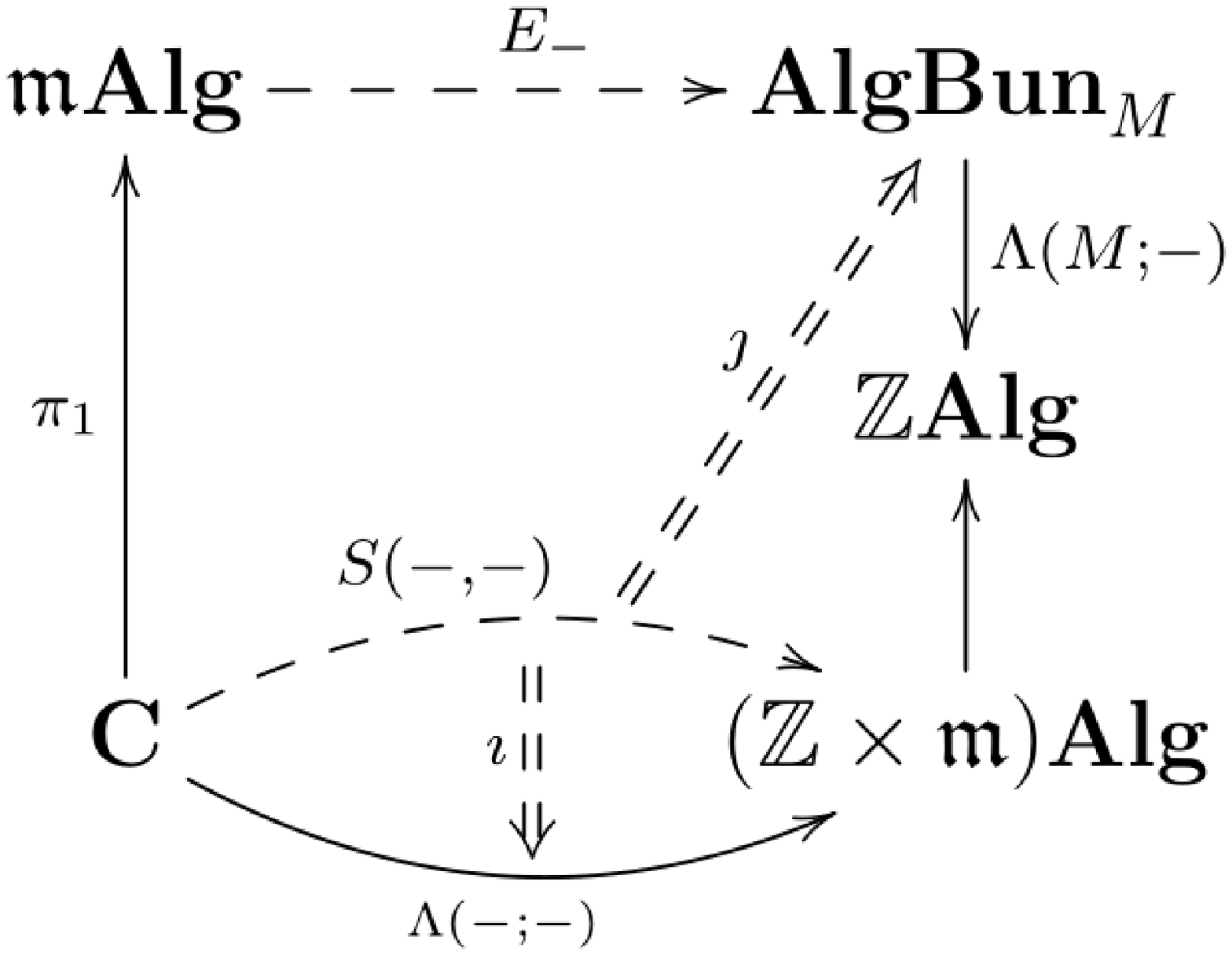} 
\par\end{center}

For a chosen bundle $P\rightarrow M$ and a $\mathfrak{m}$-graded
algebra $A\simeq A_{0}\oplus A_{1}$, we define a \emph{compatible}
\emph{graded FABS} exactly as in the last subsection. Notice that
now the natural transformation $\imath:S\Rightarrow\Lambda$ is $(\mathbb{Z}\times\mathfrak{m}$)-graded,
while $\jmath:S\Rightarrow\Lambda(M;E_{-})$ remains $\mathbb{Z}$-graded
(due precisely to the fact that $E_{A}$ may not be $\mathfrak{m}$-graded).
In particular, the diagram (\ref{invariant_FABS}) remains in $\mathbb{Z}\mathbf{AlgBun}_{M}$.
This means that the trace transformation that we need to consider
is also only $\mathbb{Z}$-graded, as previously. For every such compatible
graded FABS endowed with a trace transformation $\operatorname{tr}$
we define the corresponding\emph{ $A$-valued inhomogeneous EHP} \emph{theory
}as the classical theory whose configuration space is the collection
of all $A$-valued connections in $P$, and whose action functional
is given by 
\begin{equation}
S_{n,\Lambda}[e,\omega]=\int_{M}\operatorname{tr}(\jmath'(\pi(\alpha_{HP}))).\label{graded_EHP}
\end{equation}
Thus, up to minor modifications everything works as in the last subsection.
The crux of these ``minor modifications'' is that we can now locally
decompose $e$ and $\omega$. This allows us to get a stronger version
of Theorem \ref{theorem_A1} following the same strategy used to get
Theorem \ref{theorem_S1} from Theorem \ref{theorem_C1}. Indeed,
consider the following conditions about the vector space decomposition
$A\simeq A_{0}\oplus A_{1}$:
\begin{enumerate}
\item[(G1)] $(A_{0},*_{0})$ is a subalgebra and each $A_{0}^{j}=A^{i}\cap A_{0}$
is a $(k_{j},s_{j})$-weak solvable subspace of $A_{0}$ (recall definition
in Remark \ref{solv_modules}).
\item[(G2)] restricted to each $A^{j}$ the algebra $A$ is $(k_{j},s_{j})$-weak
solvable. 
\end{enumerate}
$\quad\;\,$We can now prove
\begin{thm}
\label{theorem_G1}Let $M$ be an $n$-dimensional spacetime, $P\rightarrow M$
be a bundle and $(A,*)$ be a $\mathfrak{m}$-graded algebra endowed
with a vector space decomposition $A_{0}\oplus A_{1}$ fulfilling
conditions \emph{(G1)} or \emph{(G2)}above. Furthermore, let $(k,s)$
be the minimum of $(k_{j},s_{j})$. If $n\geq k+s+1$, then for any
compatible FABS and any trace transformation, the corresponding inhomogeneous
$A$-valued EHP theory equals the homogeneous ones. If $n\geq k+s+3$,
then both theories are trivial. 
\end{thm}
\begin{proof}
The proofs for (G1) and (G2) are very similar, so that we will only
explain the (G1) case. Once again, since everything is linear and
grading-preserving, it is enough to prove that $\alpha_{HP}=0$ for
some representative Einstein-Hilbert form. Particularly, we can prove
for 
\[
\alpha_{HP}=\wedge_{*}^{n-2}e\wedge_{*}\Omega+\frac{\Lambda}{(n-1)!}\wedge_{*}^{n}e.
\]
Under the hypothesis, we can locally write $e=\sum_{m}e^{m}$ with
$e^{m}:TP\rightarrow A_{0}^{m}$. From condition (G1) each $A_{0}^{m}$
is a weak $(k_{m},s_{m})$-solvable subspace, so that it writes as
a sum $A_{0}^{m}=\oplus_{i}V_{i}^{m}$ of weak $(k_{m},s_{m})$-nilpotent
subspaces, which means, in particular, that we can write $e^{m}=\sum_{i}e_{i}^{m}$
locally and, therefore, $e=\sum_{m}\sum_{i}e_{i}^{m}$. Consequently,
for every $l$ we have $\wedge_{*}^{l}e=p_{l}(e_{i_{1}}^{m_{1}},...,e_{i_{l}}^{m_{l}})$
for some polynomial of degree $l$. If we now consider the minimum
$(k,s)$ (over $m$) of $(k_{m},s_{m})$, the fact that each $A_{0}^{m}$
is weakly $(k_{m},s_{m})$-nilpotent then implies $\wedge_{*}^{k+s+1}e=0$.
The remaining steps in the proof are identical to every other given
in the previous theorems. 
\end{proof}

\subsection{Graded Gravity \label{graded_gravity}}

The last section was about ``gauge theories with values in graded
algebras''. We would like to work not only with graded algebras but
in the ``full graded context'', i.e, with genuine graded gauge theories
(particularly, with graded EHP theories). A reductive $A$-valued
gauge theory is about $A$-valued connections, i.e, $1$-forms $\nabla:TP\rightarrow A$
which are reductive respective to some decomposition $A\simeq A_{0}\oplus A_{1}$.
So, ``graded gauge theories'' should be about ``$A$-valued graded
connections''. There are many approaches to formalize the notion
of ``graded connection''. For instance, we have:
\begin{enumerate}
\item Quillen superconnections \cite{superconnections_QUILLEN}, which are
defined as operators on graded vector bundles over a non-graded manifold; 
\item connections on graded manifolds in the spirit of Kostant-Berezin-Leites,
which are defined for graded principal bundles $P\rightarrow M$ over
graded manifolds \cite{connections_graded_principal_bundles};
\item connections on principal $\infty$-bundles, which can be applied in
the domain of any cohesive $\infty$-topos (particularly in the $\infty$-topos
of formal-super-smooth manifolds), where the notion of differential
cohomology can be axiomatized and a $\infty$-bundle with connection
is defined as a cocycle of such cohomology \cite{infinity_bundles_1,infinity_bundles_2,differential_cohomology}.
\end{enumerate}
Here we will not work with any of the models above. Instead, given
a$\mathfrak{m}$-graded algebra $A$, we assume that the bundle $TP$
is also $\mathfrak{m}$-graded (in the sense that it decomposes as
a sum of vector bundles $TP\simeq\oplus_{m}E^{m}$) and we define
an \emph{$A$-valued} \emph{graded connection} \emph{in $P$}\textbf{\emph{
}}\emph{of degree $l$} as a smooth $1$-form $\nabla:TP\rightarrow A$
which has degree $l$, meaning that it decomposes as a sequence of
usual vector-valued $1$-forms $\nabla^{m}:E^{m}\rightarrow A^{m+l}$.
Particularly, for us the\emph{ reductive graded connections} are those
that can be written as $\nabla=e+\omega$, where each $e$ and $\omega$
can themselves be decomposed as maps of degree $l$ 
\[
e^{m},\omega^{m}:E^{m}\rightarrow A^{m+l},\quad\text{with}\quad e^{m}+\omega^{m}=\nabla^{m}.
\]

With the notion of ``reductive graded connections of degree $l$'',
we can define ``graded EHP theories of degree $l$'' in a natural
way. Indeed, chosen a $\mathfrak{m}$-graded FABS, the corresponding
\emph{$\mathfrak{m}$-graded $A$-valued EHP theory} \emph{of degree
$l$ }is the classical theory given by the action functional (\ref{graded_EHP})
restricted to the class of graded connections of degree $l$. 
\begin{rem}
If we think of $A$ as a graded-algebra describing the infinitesimal
symmetries of the theory and if we interpret the graded structure
of the bundle $TP$ as induced by the different flavors of fundamental
objects of the theory, then it is more natural to consider the connections
of degree zero, because they will map each piece $E^{m}$ into each
corresponding subspace of infinitesimal symmetries $A^{m}$. On the
other hand, in some situations (say in the BV-BRST formalism, where
we have ghosts, anti-ghosts and anti-fields) we need to work with
``shifted symmetries'', meaning that degree $l>0$ graded connections
should also have some physical meaning. 
\end{rem}
Now, let us focus on the geometric obstructions of the $A$-valued
EHP theory that appear in the full graded context. First of all, notice
that to give a morphism $f:A'\rightarrow A$ of degree $l$ between
two graded algebras is the same as giving a zero degree morphism $f:A'\rightarrow A[-l]$,
where $A[-l]$ is the graded algebra obtained shifting $A$. Consequently,
the obstructions of a degree $l$ graded $A$-valued EHP theory are
just the obstructions of degree zero graded $A[-l]$-valued EHP theory,
so that it is enough to analyze theories of degree zero.

The core idea of the proof of Theorem \ref{theorem_G1} was to use
that $A$ is graded and the hypothesis (G1) or (G2) in order to conclude
that $\wedge_{*}^{k+s+1}e=0$. More precisely, the grading of $A$
was needed in order to write $e$ as $e\simeq\sum_{m}e^{m}$, with
$e^{m}:TP\rightarrow A_{0}^{m}$, while the ``solvability hypothesis''
(G1) or (G2) allowed us to decompose each $e^{m}$ as $e^{m}=\oplus_{i}e_{i}^{m}$,
so that $e=\oplus_{i,m}e_{i}^{m}$. The hypothesis (G1) or (G2) was
then used once again to wield $\wedge_{*}^{k+s+1}e=0$.

If we assume that $TP$ is a graded bundle, so that $TP\simeq\oplus_{m}E^{m}$,
and if $e$ has degree zero, then the only change in comparison to
the previous ``partially graded'' context is that instead of decomposing
$e$ as a sum $\sum_{m}e^{m}$, we can now write it as a genuine direct
sum $e=\oplus_{m}e^{m}$, with $e^{m}:E^{m}\rightarrow A_{0}\cap A^{m}$.
Therefore, under the same hypothesis (G1) or (G2), for theories of
\textbf{degree zero} we get exactly the same obstruction results.
Due to the argument of the last paragraph we then have the following
general obstruction result, which applies for graded theories of arbitrary
degree.
\begin{thm}
\label{theorem_G1_2}Let $M$ be an $n$-dimensional spacetime, $P\rightarrow M$
be a bundle such that $TP$ is $\mathfrak{m}$-graded and $(A,*)$
be a $\mathfrak{m}$-graded algebra endowed with a vector space decomposition
$A_{0}\oplus A_{1}$. Given $l\geq0$, assume that $A[-l]$ satisfies
condition\emph{ (G1)} or \emph{(G2)}. Furthermore, let $(k,s)$ be
the minimum of $(k_{j},s_{j})$. If $n\geq k+s+1$, then for any compatible
graded FABS and any trace transformation, the corresponding inhomogeneous
$A$-valued graded EHP theory of degree $l$ equals the homogeneous
ones. If $n\geq k+s+3$, then both theories are trivial.
\end{thm}
Assume that the $\mathfrak{m}$-graded algebra $A$ is bounded from
below, meaning that $\mathfrak{m}$ has a partial order $\leq$ and
that there is $m_{o}\in\mathfrak{m}$ such that $A^{m}\simeq0$ if
$m<m_{o}$. Consider a $\mathfrak{m}$-graded theory of degree $l$
taking values in that kind of algebra. In order to apply the last
theorem we need to verify one of conditions (G1) or (G2). Such conditions
are about $A[-l]$. This algebra has grading $A[-l]^{m}=A^{m+l}$.
Because $A$ is bounded, the first solvability condition falls on
$A^{m_{o}+l}$, i.e, \emph{no condition is needed between $m_{o}$
and $m_{o}+l$}. This is, in essence, the new phenomenon obtained
when we work with ``full graded theories'' in the sense introduced
above. As a corollary:
\begin{cor}
In the same notation of the last theorem, assume that $A$ is bounded
from below and from above, respectively in degrees $m_{o}$ and $m_{1}$.
If $l$ does not divide $m_{1}-m_{0}$, then \emph{(G1)} and \emph{(G2)}
are automatically satisfied and, therefore, for any graded bundle
$P$ the corresponding graded EHP theory of degree $l$ is trivial.
Otherwise, i.e, if $m_{1}-m_{0}=k.l$ for some $k$, then conditions
\emph{(G1)} and \emph{(G2)} need to be fulfilled exactly for $k$
terms.
\end{cor}
A particular consequence is the following:
\begin{cor}
Let $A\simeq A^{0}\oplus A^{1}$ be a $\mathbb{Z}_{2}$-graded algebra,
endowed with a vector space decomposition $A\simeq A_{0}\oplus A_{1}$,
where $A_{i}$ is not necessarily $A^{i}$ and $A_{0}$ is a subalgebra
of $A$. If $A_{0}\cap A^{1}$ is a weak $(k,s)$-solvable subspace
of $A_{0}$, then any $A$-valued $\mathbb{Z}_{2}$-graded EHP theory
of degree one over a $n$-dimensional spacetime $M$ is trivial if
$n\geq k+s+3$.
\end{cor}
Examples will be explored in next section.

\subsection{Some Speculation}

When we look at the previous obstruction theorems, all of them (except
Theorem \ref{holonomy}) were conceived as abstractions of a single
``fundamental obstruction theorem'', namely Theorem \ref{theorem_C1},
introduced in the most concrete situation: the \emph{linear/matrix
context}. Despite the fact that ``derived obstruction theorems''
hold in more abstract contexts, they are very closely related to the
first one, in that they require the same kind of hypothesis on the
underlying algebra: a\emph{ }``solvability condition''. Furthermore,
the more abstract the context is, the weaker the required ``solvability
condition'' is. Indeed, Theorem \ref{theorem_C1} (for linear EHP
theories) requires ``$(k,s)$-nil condition'', while Theorem \ref{theorem_S1}
(for gauge EHP theories) requires ``$(k,s)$-solv condition''. Furthermore,
Theorem \ref{theorem_A1} (for arbitrary $A$-valued EHP theories)
requires ``partially $(k,s)$-solv condition'', in the sense that
$A\simeq A_{0}\oplus A_{1}$, with $A_{0}$ $(k,s)$-solv, and Theorem
\ref{theorem_G1} (for arbitrary graded-valued EHP theories) requires
``locally $(k,s)$-solv condition'', meaning that each $A_{0}^{m}$
is weak $(k,s)$-solvable.

On the other hand, the same strategy used here can a piori be applied
to get geometric obstructions for any classical theory $(\operatorname{Conf},S)$
whose space of configurations $\operatorname{Conf}$ is some class
of algebra-valued smooth forms. Let us call such theories \emph{smooth
forms theories}. However, recall that our strategy here was based
in working first in the ``linear context'', where we identify a
``fundamental algebra condition''. Then, abstracting the context
we could consider weaker conditions than the ``fundamental'' one.
Therefore, in order to use this strategy in other results we need
to find a corresponding ``fundamental algebra condition''. This
leads us to speculate:
\begin{conjecture*}[roughly]
Any smooth forms theory admits a fundamental algebraic condition. 
\end{conjecture*}
For instance, in the way that it is stated, it is easy to verify that
this conjecture is true for ``polynomial smooth forms theories'',
i.e, for theories whose configuration space is $\operatorname{Conf}=\Lambda(P;A)$,
for some algebra $A$ endowed with a vector space decomposition $A\simeq A_{1}\oplus...\oplus A_{k}$,
and whose action functional is 
\[
S[\alpha_{1},...,\alpha_{k}]=\int_{M}\jmath(p_{s}(\alpha_{1},...,\alpha_{k})),
\]
where $\jmath$ is some FABS and $p_{s}$ is a polynomial of degree
$s$ in $\Lambda(P;A)$. Indeed, the desired condition is that $p_{s}$
be a PI of $A$. Now, notice that EHP theories are ``polynomial smooth
forms theories of degree'' in the above sense. Therefore, they have
a ``fundamental algebra condition'' given by the vanishing of $p_{s}$.
This is a \emph{nilpotent condition}, which is much stronger than
the \emph{$(k,s)$-solv condition} previously established. This teaches
us that the same smooth forms theory may admit two fundamental algebraic
conditions, leading us to search for the optimal one:
\begin{conjecture*}[roughly]
Any smooth forms theory admits an optimal fundamental algebraic condition. 
\end{conjecture*}
We can go further and ask if there is some kind of ``universal algebraic
condition''. More precisely, suppose a collection $\mathcal{C}$
of smooth forms theories satisfying last conjecture is given. Given
two of those theories, if the optimal algebraic condition of one is
contained in the optimal algebraic condition of the other, then both
can be simultaneously trivialized. If this is not the case, the union
of those optimal conditions will clearly trivialize them simultaneously.
But, the union of optimal conditions is not necessarily the optimal
one. This leads us to define the \emph{optimal fundamental algebraic
condition} of $\mathcal{C}$ as the weaker algebraic condition under
which each classical theory in $\mathcal{C}$ becomes trivial, and
to speculate its existence:
\begin{conjecture*}[roughly]
Any collection $\mathcal{C}$ of smooth forms theories admits an
optimal fundamental algebraic condition. 
\end{conjecture*}

\section{Examples \label{sec_examples}}

In the present section we will give realizations of the obstruction
theorems studied previously. Due to the closeness with the genuine
Lorentzian EHP theory, our focus is on the ``concrete context'',
meaning that we will give many examples of ``concrete geometries''
which realize the obstruction theorems of Section \ref{sec_obstruction}.
Even so, some examples for the ``abstract context'' of Section \ref{sec_abstract_obstructions}
will also be given.

\subsection{Linear Examples \label{linear_examples}}

We start by considering the ``fully linear'' context of Subsection
\ref{matrix_gravity}, i.e, EHP theories defined on a $G$-principal
bundle $P\rightarrow M$ respective to a group reduction $H\hookrightarrow G$,
where $G$ is a linear group and $M$ is a smooth $n$-manifold. The
objects of interest are the classical reductive Cartan connections
on $P$, i.e, pairs $(e,\omega)$ of $1$-forms such that $e$ takes
values in $\mathfrak{g}/\mathfrak{h}$ and $\omega$ takes values
in $\mathfrak{h}$. It follows from Corollary \ref{corollary_C1}
and Theorem \ref{theorem_duals} that if $n\geq6$ and 
\begin{enumerate}
\item[(E1)] $H\subset G$ is normal, with $\mathfrak{g}/\mathfrak{h}\subset\mathfrak{so}(k_{1})\oplus...\oplus\mathfrak{so}(k_{r})$,
then the corresponding gauge EHP theory is trivial; 
\item[(E2)] $\mathfrak{h}\subset\mathfrak{so}(k_{1})\oplus...\oplus\mathfrak{so}(k_{r})$,
then the geometric/algebraic dual EHP theories are trivial;
\item[(E3)] $\mathfrak{g}\subset\mathfrak{so}(k_{1})\oplus...\oplus\mathfrak{so}(k_{r})$,
then both EHP and the dual theories are trivial.
\end{enumerate}
Conditions (E2) and (E3) are immediately satisfied if we identify
$k$-groups, i.e, Lie subgroups of $SO(k_{1})\times...\times SO(k_{r})$,
or, equivalently, Lie subalgebras of $\mathfrak{so}(k_{1})\oplus...\oplus\mathfrak{so}(k_{r})$.
Indeed, if $A$ is a such subgroup, then the EHP theory for $A\hookrightarrow G_{A}$,
where $G_{A}$ is any matrix extension of $A$, obviously satisfies
(E2), so that if $n\geq6$ the dual theories are trivial. On the other
hand, if we consider EHP theories for reductions $H\hookrightarrow A$,
then (E3) is clearly satisfied and in dimension $n\geq6$ both the
dual and the actual EHP theory are trivial.

Some obvious examples of subgroups of $A\subset SO(k)$, are given
in the table below. Except for $Spin(4)\hookrightarrow O(8)$, which
arises from the exceptional isomorphism $Spin(4)\simeq SU(2)\times SU(2)$,
all the other inclusions already appeared in Subsection \ref{holonomy}.
Let us focus on condition (E2). In this case we can think of each
element of Table \ref{first_examples} as included in some $GL(k;\mathbb{R})$,
i.e, as a $G$-structure on a manifold and then as a geometry (recall
Table \ref{berger}). 
\begin{table}[H]
\begin{centering}
\begin{tabular}{|c|c|c|c|c|c|}
\hline 
$(n,k)$  & $(1,k)$  & $(2,k)$  & $(4,k)$  & $(1,7)$  & $(1,8)$\tabularnewline
\hline 
$H\subset O(n,k)$  & $SO(k)$  & $U(k)$  & $Sp(k)$  & $G_{2}$  & $Spin(4)$\tabularnewline
\hline 
\multicolumn{1}{c|}{} & $U(k)$  & $SU(k)$  & $Sp(k)\cdot Sp(1)$  &  & $Spin(7)$\tabularnewline
\cline{2-6} \cline{3-6} \cline{4-6} \cline{5-6} \cline{6-6} 
\end{tabular}
\par\end{centering}
\caption{\label{first_examples}First classical examples of geometric obstructions}
\end{table}

We could get more examples by taking finite products of arbitrary
elements in the table. In terms of geometry, this can be interpreted
as follows. Recall that a regular distribution of dimension $k$ on
an $n$-manifold can be regarded as a $G$-structure for $G=GL(k)\times GL(n-k)$.
Therefore, we can think of a product $O(k)\times O(n-k)$ as a regular
distribution of Riemannian leaves on a Riemannian manifold, $U(k)\times U(n-k)$
as a hermitan distribution, and so on.

Other special cases where condition (E2) applies are in table below.
In the first line, $O(k,k)$ is the so-called \emph{Narain group}
\cite{narain_group}, i.e, the orthogonal group of a metric with signature
$(k,k)$, whose maximal compact subgroup is $O(k)\times O(k)$. Second
line follows from an inclusion similar as $U(k)\hookrightarrow O(2k)$,
first studied by Hitchin and Gualtieri \cite{generalized_complex_PRIMEIRO,generalized _complex_thesis,generalized_complex_2},
while the remaining lines are particular cases of the previous ones.
The underlying flavors of geometry arose from the study of Type II
gravity and Type II string theory \cite{generalized_geometry_TYPE_II,generalized_geometry_TYPE_1}.
That condition (E2) applies for the second and third lines follows
from the fact that complexifying $U(k,k)\hookrightarrow O(2k,2k)$
we obtain $U(k,k)\hookrightarrow O(4k;\mathbb{C})$, as will be discussed
in the next section. 
\begin{table}[H]
\begin{centering}
\begin{tabular}{|c|c|c|}
\hline 
$G$  & $H$  & $\text{geometry}$\tabularnewline
\hline 
\hline 
$O(k,k)$  & $O(k)\times O(k)$  & Type II\tabularnewline
\hline 
$O(2k,2k)$  & $U(k,k)$  & Generalized Complex\tabularnewline
\hline 
$O(2k,2k)$  & $SU(k,k)$  & Generalized Calabi-Yau\tabularnewline
\hline 
$O(2k,2k)$  & $U(k)\times U(k)$  & Generalized Kähler\tabularnewline
\hline 
$O(2k,2k)$  & $SU(k)\times SU(k)$  & Generalized Calabi\tabularnewline
\hline 
\end{tabular}
\par\end{centering}
\caption{\label{more_examples}More examples of classical geometric obstructions}
\end{table}

About these two tables, some remarks:
\begin{enumerate}
\item Table \ref{first_examples} contains any ``classical'' flavors of
geometry, except symplectic geometry. The reason is that the symplectic
group $\operatorname{Sp}(k;\mathbb{R})$ is not contained in some
$O(r)$. But this does not mean that condition (E2) cannot be satisfied
by symplectic geometry. Indeed, generalized complex geometry contains
symplectic geometry \cite{generalized _complex_thesis}, so that Table
\ref{more_examples} implies that symplectic geometry fulfill condition
(E2). 
\item We could create a third table with ``exotic $k$-groups'', meaning
that a priori they are not related to any ``classical geometry'',
so that they describe some kind of ``exotic geometry''. For instance,
in \cite{more_subgroups_O(n)} all Lie subgroups $H\subset O(k)$
satisfying 
\[
\frac{(k-3)(k-4)}{2}+6<\dim H<\frac{(k-1)(k-2)}{2}
\]
were classified and in arbitrary dimension $k$ there are fifteen
families of them. Other exotic (rather canonical, in some sense) subgroups
that we could add are maximal tori. Indeed, both $O(2k)$ and $U(k)$
have maximal tori, say denoted by $T_{O}$ and $T_{U}$, so that the
reductions $T_{O}\hookrightarrow O(2k)$ and $T_{U}\hookrightarrow U(k)$
will satisfy condition E2. More examples of exotic subgroups to be
added are the \emph{point groups}, i.e, $H\subset\operatorname{Iso}(\mathbb{R}^{k})$
fixing at least one point. Without loss of generality we can assume
that this point is the origin, so that $H\subset O(k)$. Here we have
the symmetric group of any spherically symmetric object in $\mathbb{R}^{k}$,
such as regular polyhedrons and graphs embedded on $\mathbb{S}^{k-1}$. 
\item If we are interested only in condition (E2), then the tables above
can be enlarged by including embeddings of $O(k_{1})\times...\times O(k_{r})$
into some other larger group $\tilde{G}$. Indeed, in this condition
it only matters that $H$ is a $k$-group. Particularly, $O(k)$ admits
some exceptional embeddings (which arise from the classification of
simple Lie algebras), as in the table below \cite{exceptional_embed}.
\begin{table}[H]
\begin{centering}
\begin{tabular}{|c|c|c|c|c|c|}
\hline 
$k$  & 3  & 9  & 10  & 12  & 16\tabularnewline
\hline 
$O(k)\hookrightarrow$  & $G_{2}$  & $F_{4}$  & $E_{6}$  & $E_{7}$  & $E_{8}$\tabularnewline
\hline 
\end{tabular}
\par\end{centering}
\caption{\label{exceptional_embed} Exceptional embbedings of orthogonal groups.}
\end{table}
\end{enumerate}
Due to the above discussion, Tables \ref{first_examples} and \ref{more_examples},
as well as the possible ``table of exotic $k$-groups'', are not
only a source of examples for condition (E2), but also for condition
(E3). Indeed, we can just consider $G$ in condition (E3) as some
``$H$'' in the tables and take an arbitrary $H\subset G$. This
produces a long list of examples, because $G/H$ is a priori an arbitrary
homogeneous space subject only to the condition that $G$ is a $k$-group.
In geometric terms, if a geometry fulfills condition (E2), then any
induced ``homogeneous geometry'' fulfills (E3). On the other hand,
differently from what happens with condition (E2), Table \ref{exceptional_embed}
\emph{cannot} be used to get more examples of condition (E3). Indeed,
if $H\hookrightarrow G$ is a reduction fulfilling (E3) and $G\hookrightarrow\tilde{G}$
is an embedding, then $H\hookrightarrow\tilde{G}$ fulfills (E3) iff
it is satisfied by $G\hookrightarrow\tilde{G}$ (see diagram below).$$
\xymatrix{H\ar@/_{0.6cm}/[rrr]_{}\ar@{^{(}->}[r] & G\ar@{^{(}->}[r] & O(k_{1})\times...\times O(k_{r})\ar@{^{(}->}[r] & \tilde{G}}
$$

Let us now analyze condition (E1). First of all we notice that it
is very restrictive, because we need to work with reductions $H\hookrightarrow G$
in which $H\subset G$ is normal. For instance, in the typical situations
above, $H$ is not normal. Even so, there are two dual conditions
under which $H\hookrightarrow G$ fulfills (E1): when $G$ is a $k$-group
with $G/H\subset G$ and, dually, when $H$ is $k$-group with $G/H\subset H$.
In the first situation, we are just in condition (E3) for $G/H\hookrightarrow G$,
while in the second we are in condition (E3) for $G/H\hookrightarrow H$.
Therefore, there are not many new examples here. $\underset{\underset{\;}{\;}}{\;}$

\noindent \textbf{Conclusion.} \emph{Assume $n\geq6$}. \emph{So,
for any compatible FABS:} 
\begin{enumerate}
\item \emph{geometric/algebraic dual EHP is trivial in each geometry modeled
by tables above;} 
\item \emph{the actual gauge EHP is trivial in each homogeneous geometry
associated to the first two tables above.} 
\end{enumerate}

\subsection{Extended-Linear Examples \label{extended_linear_examples}}

In the last subsection we constructed explicit examples of geometric
obstructions for linear EHP theories and their duals. Let us now consider
extended-linear theories. This means that we will work with bundles
$P\rightarrow M$ structured over $\mathbb{R}^{k}\rtimes H$, where
$H$ is a linear group and $M$ is a $n$-dimensional smooth manifold.
Furthermore, we will take into account reductive connections for the
group reduction $H\hookrightarrow\mathbb{R}^{k}\rtimes H$. Our obstruction
theorem is now Corollary \ref{corollary_S0}, which implies that if
$n\geq6$ and
\begin{itemize}
\item $\mathfrak{h}\subset\mathfrak{so}(k_{1})\oplus...\oplus\mathfrak{so}(k_{r})$,
then EHP action and its duals are trivial.
\end{itemize}
Notice that this is exactly the same condition as (E2). This means
that \emph{if we are working in the extended-linear context, in dimension
$n\geq6$ EHP theory cannot be realized in any geometry of the last
subsection.}
\begin{rem}
Recall that it is in the extended-linear context, with $P=FM$, that
``abstract EHP theory'' becomes very close to the concrete one,
in that the geometries of the last subsection really describe geometries
(in the most classical sense, in terms of tensors) on the manifold
$M$. Therefore, the fact that EHP is trivial in a lot of situations
is a strong manifestation that the geometry of gravity is very rigid,
especially in higher dimension.
\end{rem}

\subsection{Cayley-Dickson Examples}

In the last two sections we presented obstructions to the realization
of linear/extended-linear (dual) EHP theories in commonly studied
geometries, as well as in some ``exotic'' geometries. Here we would
like to show that there are also obstructions for less studied (and
some never studied in detail) kinds of geometry. So, for now, let
$(A,*)$ be an arbitrary $\mathbb{R}$-algebra and let $\operatorname{Mat}_{k\times l}(A)$
be the $\mathbb{R}$-module of $k\times l$ matrices with coefficients
in $A$. The multiplication of $A$ induces a corresponding multiplication
\[
\cdot:\operatorname{Mat}_{k\times l}(A)\times\operatorname{Mat}_{l\times m}(A)\rightarrow\operatorname{Mat}_{k\times l}(A).
\]
In particular, for $k=0=m$ we see that for every $l$ we have a bilinear
map 
\[
b:A^{l}\times A^{l}\rightarrow A,\quad\text{given by}\quad b(x,y)=x_{1}*y_{1}+...+x_{l}*y_{l}.
\]
where we used the identifications 
\[
\operatorname{Mat}_{0\times l}(A)\simeq A^{l}\simeq\operatorname{Mat}_{l\times0}(A)\quad\text{and}\quad\operatorname{Mat}_{0\times0}(A)\simeq A.
\]

\begin{rem}
The bilinear map above is symmetric iff the algebra $A$ is commutative.
Furthermore, its non-degeneracy depends on whether $A$ has zero divisors
or not, and a priori it is not possible to ask about its positive
definiteness, because it takes values in $A$ and not in $\mathbb{R}$.
Thus, it is \textbf{not} an inner product in $\operatorname{Mat}_{0\times l}(A)$.
But, if we choose $A=\mathbb{R}$, then it is the standard inner product
of $\mathbb{R}^{l}$. Now, assume that $A$ is endowed with an involution
$\overline{(-)}:A\rightarrow A$. In this case, it can be combined
with $b$ to get a sesquilinear map $s$ in $A^{l}$, as in the diagram
below. $$
\xymatrix{\ar@/_{0.5cm}/[rrr]_{s}A^{l}\times A^{l}\ar[rr]^{\overline{(-)}^{l}\times id}&&A^{l}\times A^{l}\ar[r]^b&A}
$$Explicitly, 
\[
s(x,y)=\overline{x}_{1}*y_{1}+...+\overline{x}_{l}*y_{l}.
\]
We can now consider the subspace of all $l\times l$ matrices $M\in\operatorname{Mat}_{l\times l}(A)$
with coefficients in $A$ which preserve the sesquilinear form $s$,
in the sense that $s(Mx,My)=s(x,y).$ We say that these are the \emph{unitary
matrices in $A$}, \emph{respective to the involution $s$ induced
by the involution $\overline{(-)}$, }and we denote this set by $U(k;A)$.
If the involution is trivial (i.e, the identity map), then call them
the \emph{orthogonal matrices} \emph{in $A$}, writing $O(k;A)$ to
denote the corresponding space.
\end{rem}
\begin{example}
If we consider $A=\mathbb{R},\mathbb{C}$ and the trivial involution
we get, respectively, the real and the complex orthogonal groups $O(k)$
and $O(k;\mathbb{C})$. If we consider the canonical involution of
$\mathbb{C}$, we get the unitary group $U(k)$, while if we consider
$\mathbb{H}$ with its canonical involution we get the quaternionic
unitary group $U(k;\mathbb{H})$, which already have appeared in Tables
\ref{berger} and \ref{first_examples} with the more used notation
$Sp(k)$.
\end{example}
The last example leads us to compare different unitary groups of involutive
algebras which are related by the Cayley-Dickson construction discussed
in Example \ref{CD}. Indeed, recall that this construction takes
an involutive algebra $A$ and gives another involutive algebra $\operatorname{CD}(A)$
with weakened PI's. As an $R$-module, the newer algebra is given
by a sum $A\oplus A$ of ``real'' and ``imaginary'' parts. We
have an inclusion $A\hookrightarrow\operatorname{CD}(A)$, obtained
by regarding $A$ as the real part, which induces an inclusion into
the corresponding unitary groups $U(k;A)\hookrightarrow U(k;\operatorname{CD}(A))$.
This inclusion can be extended in the following way 
\[
U(k;A)\hookrightarrow U(k;\operatorname{CD}(A))\hookrightarrow U(2k;A),
\]
defined by setting the ``real part'' and the ``imaginary part''
as diagonal block matrices. Iterating we see that for every $k$ and
every $l$ there exists an inclusion 
\begin{equation}
U(k;\operatorname{CD}^{l}(A))\hookrightarrow U(2^{l}k;A).\label{inclusion_unitary}
\end{equation}

\begin{example}
If we start with $A=\mathbb{R}$ endowed with the trivial involution,
the first inclusion is $U(k)\hookrightarrow O(2k)$, which describes
complex geometry. The next iteration gives $Sp(k)\hookrightarrow O(4k)$,
which is quaternionic geometry. These are precisely the first lines
of Table \ref{first_examples}, but we can continue getting octonionic
geometry, sedenionic geometry, etc. 
\end{example}
Recalling that geometry can be regarded as the inclusion of Lie groups
$H\hookrightarrow G$, with the previous examples in ours minds, the
idea is to think of inclusion (\ref{inclusion_unitary}) as some flavor
of geometry, which we could name \emph{Cayley-Dickson geometry in
$A$, of order }$l$. But, this makes sense only if the unitary groups
in (\ref{inclusion_unitary}) are Lie groups. We notice that this
is the case at least when $A$ is finite-dimensional, as it will be
sketched now.

The idea is to reproduce the proof that the standard unitary/orthogonal
groups $U(k)$ and $O(k)$ are Lie groups. The starting point is to
notice that the involution of $A$ induces an involution in $\operatorname{Mat}_{k\times k}(A)$,
defined by the following composition, where $t$ is the transposition
map.$$
\xymatrix{\ar@/_{0.5cm}/_{(-)^{\dagger}}[rr]\operatorname{Mat}_{k\times k}(A)\ar[r]^{t} & \operatorname{Mat}_{k\times k}(A)\ar[r]^{\overline{(-)}} & \operatorname{Mat}_{k\times k}(A)}
$$We then notice that a $k\times k$ matrix $M$ with coefficients in
$A$ is unitary respective to $s$ iff $M$ is invertible with $M^{-1}=M^{\dagger}$,
i.e, iff $M\cdot M^{\dagger}=1_{k}=M^{\dagger}\cdot M$. Since $\overline{(-)}:A\rightarrow A$
is an algebra morphism, we have the usual property $(M\cdot N)^{\dagger}=M^{\dagger}\cdot N^{\dagger}$,
allowing us to characterize the unitary matrices as those satisfying
$M\cdot M^{\dagger}=1_{k}$. Therefore, defining the map 
\[
f:\operatorname{Mat}_{k\times k}(A)\rightarrow\operatorname{Mat}_{k\times k}(A)\quad\text{by}\quad f(M)=M\cdot M^{\dagger},
\]
in order to proof that $U(k;A)$ is a Lie group it is enough to verify
that the above map is, in some sense, a submersion (which will give
the smooth structure) and that the multiplication and inversion maps
are ``smooth''. It is at this point that we require that $A$ to
be finite-dimensional.

Now, we can return to the context of Geometric Obstruction Theory.
Taking $A=\mathbb{R}$, for every $k$ and $l$ the discussion above
gives a sequence of Lie group inclusions$$
\xymatrix{U(k;\operatorname{CD}^{l}(\mathbb{R}))\ar[r] & \cdots\ar[r] & U(2^{l-3}k;\mathbb{O})\ar[r] & U(2^{l-2}k;\mathbb{H})\ar[r] & U(2^{l-1}k;\mathbb{C})\ar[r] & O(2^{l}k)}
$$which imply that (in spacetime dimension $n\geq6$) EHP theory and
its duals are trivial for each of these geometries (because they satisfy
(E2)). In sum, Table \ref{first_examples} can be augmented to include
Cayley-Dickson geometries in $\mathbb{R}$ of arbitrary order, as
shown below.

\begin{table}[H]
\begin{centering}
\begin{tabular}{|c|c|c|}
\hline 
$G$  & $H$  & $\text{geometry}$\tabularnewline
\hline 
\hline 
$\mathbb{R}^{k}\rtimes H$  & $O(k)$  & $\text{Riemannian}$\tabularnewline
\hline 
$\mathbb{R}^{2k}\rtimes H$  & $U(k;\mathbb{C})$  & $\text{hermitean}$\tabularnewline
\hline 
$\mathbb{R}^{4k}\rtimes H$  & $U(k;\mathbb{H})$  & $\text{quaternionic}$\tabularnewline
\hline 
$\mathbb{R}^{8k}\rtimes H$  & $U(k;\mathbb{O})$  & $\text{octonionic}$\tabularnewline
\hline 
$\mathbb{R}^{16k}\rtimes H$  & $U(k;\mathbb{S})$  & $\text{sedenionic}$\tabularnewline
\hline 
$\vdots$  & $\vdots$  & $\vdots$\tabularnewline
\hline 
$\mathbb{R}^{2^{l}k}\rtimes H$  & $U(k;\operatorname{CD}^{l}(\mathbb{R}))$  & real $\text{CD of order \ensuremath{l}}$ \tabularnewline
\hline 
$\vdots$  & $\vdots$  & $\vdots$\tabularnewline
\hline 
\end{tabular}
\par\end{centering}
\caption{\label{extended_examples}Extended examples of concrete geometric
obstructions}
\end{table}

In Subsections \ref{linear_examples} and \ref{extended_linear_examples}
we showed that EHP theory cannot be realized in the most classical
flavors of geometry. These, however, constitute a \emph{finite} amount.
As a corollary of the construction above, we now know that EHP theory
is actually trivial in an\emph{ infinite }number of geometries:
\begin{cor}
In dimension $n\geq6$, the EHP theory and its duals cannot be realized
in an infinite number of geometries. 
\end{cor}
Table \ref{extended_examples} can, in turn, be extended in three
different directions:
\begin{enumerate}
\item \emph{\uline{by adding distinguished subgroups}}. For every involutive
algebra $A$ we can define the subgroup $SU(k;A)\subset U(k;A)$ of
those matrices whose determinant equals the identity $1\in A$. When
$A$ is finite-dimensional, it will be a Lie subgroup, allowing us
to include $SU(k;\operatorname{CD}^{l}(\mathbb{R}))$ for every $l\geq0$
in Table \ref{extended_examples}. This means that if a geometry belongs
to the ``table of obstructions'', then its ``oriented version''
belongs too;
\item \emph{\uline{by replacing the base field}}. Notice that condition
$\mathfrak{h}\subset\mathfrak{so}(k_{1})\oplus...\oplus\mathfrak{so}(k_{r})$
was used above only in order to have $\alpha\curlywedge\alpha=0$
for every even-degree $\mathfrak{h}$-valued form, i.e, to ensure
that $\mathfrak{h}$ is $(k,1)$-nil for each $k$ even. In turn,
the condition $\alpha\curlywedge\alpha=0$ is satisfied exactly because
$\mathfrak{so}(n)$ is an algebra of skew-symmetric matrices. But
this remains valid for the Lie algebra $\mathfrak{so}(n;A)$ of $SO(n;A)$,
independently of the involutive algebra. This means that we can replace
$\mathbb{R}$ by any algebra $A$ in Table \ref{extended_examples}.
New interesting examples that arise from this fact are the following.
For given $p,q>0$, there is no $k$ such that $O(p,q)\subset O(k)$,
so that a priori we cannot add $O(p,q)$ to Table \ref{extended_examples}.
On the other hand, after complexification, i.e, after replacing $\mathbb{R}$
by $\mathbb{C}$ we have $O(p,q)\otimes_{\mathbb{R}}\mathbb{C}\simeq O(p+q)\otimes_{\mathbb{R}}\mathbb{C}$
for every $p,q\geq0$, so that they immediately enter in the ``obstruction
table''. This means that we can add to Table \ref{extended_examples}
``complex semi-Riemannian geometry'' and, similarly, ``quaternionic
semi-Hermitean geometry'', and so on.
\item \emph{\uline{by making use of Lie theory}}. The condition that
we need is on the\textbf{ Lie algebra} level. It happens that in general
there are many groups with the same algebra. Therefore, once we find
a Lie group whose algebra fulfills what we need, we can automatically
add to our ``table of obstructions'' every other Lie group that
induces the same algebra. In particular, we can double the size of
our current table by adding the universal cover (when it exists) of
each group. For instance, in Tables \ref{first_examples} and \ref{extended_examples}
the spin groups $\operatorname{Spin}(4)$ and $\operatorname{Spin}(7)$
were added due to exceptional isomorphisms on the level of \textbf{Lie
groups}. Now, noticing that for $k>2$ (in particular for $k=4,7$)
$\operatorname{Spin}(k)$ is the universal cover of $SO(k)$, we can
automatically add all theses spin groups to our table, meaning that
the dual EHP cannot be realized in ``spin geometry''. Similarly,
we can add the universal coverings of the (connected component at
the identity of) $U(k;\operatorname{CD}^{l}(\mathbb{R}))$. We can
also add the universal cover of the symplectic group $Sp(k;\mathbb{R})$,
usually known as the \emph{metapletic group}. Therefore, dual EHP
cannot be realized in ``metapletic geometry'' too. 
\end{enumerate}
Up to this point we gave examples which extend Table \ref{first_examples}.
We notice, however, that Table \ref{more_examples} can also be extended.
Indeed, for every $A$ we have the inclusion $U(k,k;\operatorname{CD}(A))\subset O(2k;2k,A)$,
so that by iteration we get 
\begin{equation}
U(k,k;\operatorname{CD}^{l}(A))\subset O(2^{l}k;2^{l}k,A).\label{generalized_CD}
\end{equation}

For $A=\mathbb{R}$ and $l=1$ this model generalized complex geometry,
leading us to say that the inclusion above models \emph{``generalized
Cayley-Dickson geometry in $A$ of degree $l$}''. For instance,
if we take $A=\mathbb{R}$ and $l=2$ this becomes generalized quaternionic
geometry, which is a poorly studied theory, started with the works
\cite{generalized_quaternionic_INITIAL,generalized_quaternionic_1}
. For $l=3,4,...$ it should be ``generalized octonionic geometry'',
``generalized sedenionic geometry'', and so on. The authors are
unaware of the existence of substantial works on these theories.

When tensoring inclusion (\ref{generalized_CD}) with $\operatorname{CD}^{l}(A)$
we get 
\[
U(k,k;\operatorname{CD}^{l}(A))\subset O(2^{l}k;2^{l}k,\operatorname{CD}^{l}(A))\simeq O(2^{l}k+2^{l}k;\operatorname{CD}^{l}(A)),
\]
so that condition 2 above implies that, independently of the present
development of abstract generalized Cayley-Dickson geometries, EHP
theory (and its duals) is trivial in each of them.$\underset{\underset{\;}{\;}}{\;}$

\noindent \textbf{Conclusion:}\emph{ In dimension $n\geq6$ and for
any FABS, extendend-linear EHP theory and its duals cannot be realized
in a lot of geometries, which include:}
\begin{enumerate}
\item \emph{real Cayley-Dickson geometries of all orders, their ``connected
components of the identity'' and their universal covering geometries;}
\item \emph{generalized Cayley-Dickson geometries of all orders;} 
\item \emph{the exceptional geometry $G_{2}$.}
\end{enumerate}

\subsection{Partially Abstract Examples}

The last examples were considered ``concrete'' because the algebras
underlying them are matrix algebras over the real numbers. So, in
order to give non-concrete examples it is enough to work with algebras
which are not matrix algebras ``and/or'' which have coefficient
rings other than $\mathbb{R}$. Let us call the ``or'' examples
\emph{partially abstract examples}.

The interesting part of the partially abstract situations is that
we can give abstract examples of obstructions without using the abstract
theorems of Section \ref{sec_abstract_obstructions}. For instance,
some obstructions for non-real matrix algebras where actually given
in last subsection in ``\emph{\uline{by replacing the base-field}}''.
Indeed, there we considered situations in which the new field $\mathbb{K}$
arises as an extension of $\mathbb{R}$ and the matrix algebra $A(k;\mathbb{K})$
in consideration is indeed a scalar extension $A(k;\mathbb{K})=A(k;\mathbb{R})\otimes_{\mathbb{R}}\mathbb{K}$.
But, since $\mathbb{R}$ is characteristic zero, it follows that every
extension $\mathbb{K}\supset\mathbb{R}$ is also characteristic zero,
allowing us to ask: \emph{can we give examples when $\mathbb{K}$
has prime characteristic?}

In \cite{maximal_tori} necessary and sufficient conditions were given
under which arbitrary orthogonal groups $O(q,\mathbb{K})$, where
$q:V\rightarrow\mathbb{K}$ is a positive-definite quadratic form
on a finite-dimension $\mathbb{K}$-space $q$ and $\operatorname{ch}(\mathbb{K})\neq2$,
admit an embedded maximal torus $\mathbf{T}(q;\mathbb{K})$. Independently
of the quadratic space $(V,q)$, the corresponding orthogonal group
is a Lie group and, as in the real case, the algebra $\mathfrak{o}(q;\mathbb{K})$
is $(2,1)$-nil. Therefore, under the conditions of \cite{maximal_tori},\emph{
EHP theory cannot be realized in any of these ``toroidal geometries''
}$\mathbf{T}(q;\mathbb{K})\hookrightarrow O(q;\mathbb{K})$. Notice
that for $\mathbb{K=\mathbb{R}}$ the same result appeared in \ref{linear_examples}
as a ``exotic linear example''.

\subsection{Fully Abstract Examples}

Finally, we give ``abstract examples'' of geometries in which EHP
theory cannot be realized. First we deal with algebra-valued geometry,
with graded geometry considered in the sequence. This means that we
work with a bundle $P\rightarrow M$ endowed with $A$-valued connections
$\nabla:TP\rightarrow A$, where $A$ is a graded a$\mathfrak{m}$-graded
algebra $A\simeq\oplus_{m}A_{m}$. The obstruction theorems are now
Theorems \ref{theorem_A1}, \ref{theorem_G1} and \ref{theorem_G1_2}.
In summary, if
\begin{enumerate}
\item[(F1)] \emph{\label{condition_F1}the algebra $A$ admits a vector space
decomposition $A\simeq A_{0}\oplus A_{1}$, where $A_{0}$ is a subalgebra
such that each $A_{0}^{m}=A_{0}\cap A^{m}$ is a weak $(k_{m},s_{m})$-solvable
subspace}, then EHP theory is trivial in dimension $n\geq k+s+1$,
where $(k,s)=\min(k_{m},s_{m})$.
\end{enumerate}
$\quad\;\,$The most basic examples are those for $\mathfrak{m}=0$
and $A_{1}=0$, i.e, the non-graded setting with $A$ itself $(k,s)$-solv.
As discussed in Subsections \ref{nil_algebras} and \ref{solv_algebras},
there are many natural examples of $(k,s)$-solv algebras, e.g, any
Lie algebra is $(k,1)$ for any even $k$. As a consequence,\emph{
in dimension $n\geq4$, any Lie algebra valued EHP theory is trivial.
}In particular, EHP theories with values in the Poincaré group $\mathfrak{iso}(n-1,1)$
are trivial. But, $\mathfrak{iso}(n-1,1)$-valued EHP theory is just
classical EHP theory, thus we conclude that General Relativity does
not makes sense in dimension $n\geq4$, which is absurd. We made a
similar mistake in Example \ref{example_mistake}: while the classical
EHP theory and the $\mathfrak{iso}(n-1,1)$-valued EHP theory take
values in the same \textbf{vector space}, at the same time that their
action functionals have the same shape, the \textbf{algebra} (and,
therefore, its properties) used to define the corresponding wedge
product is \emph{totally different}. Indeed, in the discussed cases,
the exterior products $\curlywedge_{\rtimes}$ and $[\wedge]$, respectively.

If we now allow $\mathfrak{m}$ to be nontrivial, but with $A_{1}=0$,
then condition (\ref{condition_F1}) is satisfied if each $A^{m}$
is a weak $(k_{m},s_{m})$-solvable subespace. In particular, it remains
satisfied if $A$ is itself $(k,s)$-solv. It happens that not only
Lie algebras are solv $(2,1)$-nil, but also a class of graded Lie
algebras. Consequently, EHP theories are also trivial in the domain
of graded Lie algebras. One can generalize even more thinking in EHP
theories with values in Lie superalgebras and in graded Lie superalgebras.
Indeed, a \emph{$\mathfrak{m}$-graded Lie superalgebra $\mathfrak{g}$}
is just a $\mathfrak{m}$-graded Lie algebra whose underlying PI's
(i.e, skew-commutativity and Jacobi identity) hold in the graded sense.
Particularly, this means that the $\mathbb{Z}_{2}$-grading writes
$\mathfrak{g}\simeq\mathfrak{g}^{0}\oplus\mathfrak{g}^{1}$, with
$\mathfrak{g}^{0}$ a $\mathfrak{m}$-graded Lie algebra and, therefore,
$(2,1)$-nil. It then follows that $\mathfrak{g}$ satisfies condition
F1. Summarizing, \emph{EHP theories cannot be realized in any ``Lie
algebraic'' context}.

Graded Lie Superalgebras are the first examples of algebras satisfying
(F1) with $A_{1}\neq0$, but they are far from being the only one.
Indeed, when we look at a decomposition $A\simeq A_{0}\oplus A_{1}$,
where $A_{0}$ is a subalgebra, it is inevitable to think of $A$
as an extension of $A_{1}$ by $A_{0}$, meaning that we have an exact
sequence as shown below. If $A_{0}$ is $(k,s)$-solv, then (F1) holds.
This can be interpreted as follows: \emph{ suppose that we encountered
an algebra $A_{1}$ such that EHP is not trivial there. So, EHP theory
will be trivial in any (splitting) extension of $A_{1}$ by a $(k,s)$-solv
algebra}.$$
\xymatrix{0\ar[r] & A_{0}\ar[r] & A\ar[r] & A_{1}\ar[r] & 0}
$$

In particular, because $(\mathbb{R}^{k},+)$ is abelian and, therefore,
$(1,1)$-nil, any algebra extension by $\mathbb{R}^{k}$ will produce
a context in which EHP theory is trivial. For instance, recall the
extended-linear context, which was obtained taking splitting extensions
of matrix algebras by $\mathbb{R}^{k}$. Thus, EHP theory is trivial
in every extended-linear context, so that classical EHP is trivial.
This implies tha GR is trivial, showing that \emph{we made another
mistake}. Once again, \emph{the mistake resides on the underlying
``wedge products'': the wedge product induced on an algebra extension
by the abelian group $(\mathbb{R}^{k},+)$ is }\textbf{\emph{not}}\emph{
the wedge product $\curlywedge_{\rtimes}$ studied in Subsection \ref{algebra_extensions}}.
Indeed, $\curlywedge_{\rtimes}$ does\textbf{ not} take the abelian
structure of $\mathbb{R}^{k}$ into account.

Another example about extensions is as follows: the paragraph above
shows that abelian extensions of nontrivial EHP theories are trivial,
but what about ``super extensions''? They remain trivial. Indeed,
being a Lie superalgebra, the translational superalgebra $\mathbb{R}^{k\vert l}$
(the cartesian superspace, regarded as a Lie superalgebra) is $(2,1)$-nil,
so that any algebra extension by it is trivial.

Finally, let us say that we can also consider ``fully exotic abstract
examples'' meaning geometries modeled by algebras fulfilling (F1),
but that have no physical meaning. Just to mention, in \cite{exotic_nil_1}
the authors build ``mathematically exotic'' examples of nilpotent
algebras, which fulfill (F1). If, for any reason, one tries to model
gravity as EHP with values in those algebras, one will find a trivial
theory.

\section{Conclusion \label{sec_conclusion}}

$\quad\;\,$In this notes, based on \cite{EHP_eu}, we considered
EHP action functional in different contexts and we gave obstructions
to realize gravity (modeled by these EHP actions) in several geometries.
In particular, we showed that EHP cannot be realized in almost all
``classical geometries'', including Riemannian geometry, hermitean
geometry, Kähler geometry, generalized complex geometry and many extensions
of them, as well as some exceptional geometries, such as $G_{2}$-geometry.
We also introduced the notions of geometric/algebraic duals of an
EHP theory and we have show that many obstructions also affect them.
A physical understanding of these ``dual theories'' is desirable.

In this process of finding geometric obstructions for EHP theories,
we identified a ``general obstruction'', corresponding to a ``solvability
condition'' on the underlying algebra, which led us to speculate
on the existence of a ``general obstruction'' for each gauge theory
and of a ``universal obstruction'', unifying all such ``general
obstructions''.

The speculations and the obstruction theorems developed here are independent
of the choice of a \emph{Functional Algebra Bundle System}, a concept
conceived in the present work, which codifies the data necessary to
pass from $A$-valued forms on a bundle $P\rightarrow M$ to bundle-valued
forms on $M$. A specific study of the category of these objects is
desirable (we have some work in developing stage \cite{FABS}). For
instance, which kind of limits/colimits exist in such category? It
has at least initial objects (corresponding to universal FABS)?

We believe that the main contribution of the present work is to clarify
that \emph{when working with classical theories defined by algebra-valued
differential forms, one needs to be very careful about the ``wedge
products'' used in the action funcional, given that the properties
of the underlying algebra deeply affect the properties of the corresponding
wedge product}.

\end{document}